\documentclass[11pt,twoside]{article}
\usepackage{amscd,amsmath,amstext,amsfonts,amsbsy,amssymb,amsthm}
\usepackage{multicol,graphicx,enumerate}
\usepackage{array,multirow,longtable,slashbox}
\usepackage{fancyhdr,booktabs}
\usepackage[unicode,colorlinks,linkcolor = blue,citecolor = blue]{hyperref}
\usepackage{color,alltt}
\usepackage{authblk}
\usepackage[ddmmyyyy,hhmmss]{datetime}
\usepackage{rotating}
\usepackage[authoryear,round,longnamesfirst]{natbib}
\usepackage[margin = 3cm]{geometry}
\usepackage[width=.8\textwidth, font = small, labelfont = bf, labelsep = endash, textfont = it]{caption}
\usepackage{tikz}
\usetikzlibrary{shapes,arrows,arrows.meta}

\newcommand{\I}{\mathrm{I}}
\newcommand{\E}{\mathbb{E}}
\newcommand{\Var}{\mathbb{V}\mathrm{ar}}

\numberwithin{equation}{section}

\newtheorem{theorem}{{\bf Theorem}}[section]
\theoremstyle{definition} 
\theoremstyle{plain}

\theoremstyle{definition} 

\title{Improving estimation of the volume under the ROC surface when data are missing not at random}

\author[1]{Duc-Khanh To \thanks{duckhanh.to@unipd.it}}
\author[1]{Gianfranco Adimari \thanks{gianfranco.adimari@unipd.it}}
\author[2]{Monica Chiogna \thanks{monica.chiogna2@unibo.it}}
\affil[1]{Department of Statistical Sciences, University of Padova, }
\affil[ ]{Via C. Battisti, 241; I-35121 Padova, Italy}
\affil[2]{Department of Statistical Sciences ``Paolo Fortunati'', University of Bologna, }
\affil[ ]{Via Belle Arti, 41; 40126 Bologna, Italy}
\date{ }

\providecommand{\keywords}[1]{{\small {\sc{Key words:}} Diagnostic test; Instrumental variable; Mean score equation; Nonignorable missing data mechanism; ROC analysis; Verification bias.} {\small #1}}

\begin{document}

\maketitle
\vspace{-1.2 cm}

\begin{abstract}
In this paper, we propose a mean score equation-based approach to estimate the the volume under the receiving operating characteristic (ROC) surface (VUS) of a diagnostic test, under nonignorable (NI) verification bias. The proposed approach involves a parametric regression model for the verification process, which accommodates for possible NI missingness in the disease status of sample subjects, and may use instrumental variables, which help avoid possible identifiability problems. In order to solve the mean score equation derived by the chosen verification model, we preliminarily need to estimate the parameters of a model for the disease process, but its specification is required only for verified subjects under study. Then, by using the estimated verification and disease probabilities, we obtain four verification bias-corrected VUS estimators, which are alternative to those recently proposed by \citet{toduc2019vus}, based on a full likelihood approach.  Consistency and asymptotic normality of the new estimators are established. Simulation experiments are conducted to evaluate  their finite sample performances,  and an application to a dataset from a research on epithelial ovarian cancer is presented.
\end{abstract}


\section{Introduction and background}
\label{sec:intro}
The evaluation of the accuracy of diagnostic tests (or biomarkers) is an increasingly relevant issue in modern medicine. Usually, this issue is solved by means of a receiver operating characteristic (ROC) analysis. In particular, the volume under the ROC surface  (VUS) is often used for measuring the overall accuracy of a diagnostic test when the possible disease status belongs to one of three ordered categories.  Under correct ordering, values of VUS vary from 1/6, suggesting that the test is no better than chance alone, to 1, that implies a perfect test, i.e. a test that perfectly discriminates among the three categories \citep{scu}. 

In medical studies, the VUS of a new test is typically estimated through a sample of measurements obtained by some suitable sample of patients, for which the true disease status is  assessed by means of a gold standard (GS) test. However, in many cases, due to the expensiveness and/or invasiveness of the GS test, only a subset of patients undergoes disease verification. In such situations, statistical inference based only on verified subjects is typically biased. This bias is known as verification bias.

In order to correct for verification bias, the researchers must formulate some assumptions about the selection mechanism for the disease verification. When the decision to send a subject to verification may be directly based on the presumed subject's disease status, or, more generally, the selection mechanism may depend on some unobserved covariates related to disease, the missing data mechanism is called nonignorable (NI). 

Generally speaking, NI data missingness is a challenge for inference, and, specifically, the issue of correcting for NI verification bias in ROC surface analysis is still scarcely considered in the statistical literature. To the best of our knowledge, only \citet{zhang2018estimation} and \citet{toduc2019vus} gave a contribution in this direction. In particular, \citet{toduc2019vus} extended the model discussed in \citet{liu} to match the case of three-category disease status, and proposed a likelihood-based approach to derive four bias--corrected estimators of the VUS, namely, full imputation (FI), mean score imputation (MSI), 
inverse probability weighting (IPW) and  pseudo doubly robust (PDR) estimators. 

Suppose we need to evaluate the predictive ability of a new continuous diagnostic test in a context where the disease status of a patient can be described by three ordered categories, ``non--diseased'', ``intermediate'' and ``diseased'', say. Consider a sample of $n$ subjects and let $T$, ${\mathcal{D}}$ and $\mathbf{A}$ denote the test result, the disease status and a vector of covariates for each subject, respectively.  Hereafter, we assume, without loss of generality, that test values are positively associated with the severity degree of the disease. The disease status ${\mathcal{D}}$ can be modeled as a trinomial random vector ${\mathcal{D}} = (D_1, D_2, D_3)^\top$, such that $D_k$ is a Bernoulli random variable having mean $\theta_k = \Pr(D_k = 1)$ where $\theta_1 + \theta_2 + \theta_3 = 1$. Hence,  $\theta_k$ represents the probability that a generic subject, classified according to its disease status, belong to the class $k$.  The accuracy of $T$ is measured by its VUS, which is defined as 
\begin{equation}
\mu = \frac{\E \left(\I_{i\ell r}D_{1i}D_{2\ell}D_{3r}\right)}{\E \left(D_{1i}D_{2\ell}D_{3r}\right)}, \label{org:vus}
\end{equation}
where the indices $i$, $\ell$, $r$ refer to three different subjects, $\I_{i\ell r} = \I(T_i < T_\ell < T_r) + \frac{1}{2} \I(T_i < T_\ell = T_r) + \frac{1}{2}\I(T_i = T_\ell < T_r) + \frac{1}{6}\I(T_i = T_\ell = T_r)$ and $\I(\cdot)$ is the indicator function \citep{nak:04}. When the disease status $\mathcal{D}$ is available for all subjects, a natural nonparametric estimator of $\mu$ is given by
\begin{equation}
\widehat{\mu}_{\mathrm{NP}} = \frac{\sum\limits_{i=1}^{n}\sum\limits_{\ell = 1, \ell \ne i}^{n} \sum\limits_{\stackrel{r = 1}{r \ne \ell, r \ne i}}^{n}\I_{i\ell r}D_{1i}D_{2\ell}D_{3r}}{\sum\limits_{i=1}^{n}\sum\limits_{\ell = 1, \ell \ne i}^{n} \sum\limits_{\stackrel{r = 1}{r \ne \ell, r \ne i}}^{n} D_{1i}D_{2\ell}D_{3r}}.
\label{nonp:vus}
\end{equation}
Suppose now that the disease status ${\mathcal{D}}$ is missing for a subset of patients in the study.
Let $V_i$ be the verification status for the $i$-th subject in study, that is $V_i = 1$ if $\mathcal{D}_i$ is observed and $V_i = 0$ otherwise. To account for a possible nonignorable missing data mechanism for the disease status, \citet{toduc2019vus} considered the verification model 
\begin{equation}
\pi = \Pr(V = 1|T, \mathbf{A}, D_1, D_2) = \frac{\exp\left(\beta_0 + \beta_1 T + \boldsymbol{\beta}_2 ^\top \mathbf{A} + \lambda_1 D_1 + \lambda_2 D_2 \right)}{1 + \exp\left(\beta_0 + \beta_1 T + \boldsymbol{\beta}_2 ^\top \mathbf{A} + \lambda_1 D_1 + \lambda_2 D_2 \right)},
\label{veri:model:1}
\end{equation}
where the parameters $\lambda_1$ and $\lambda_2$ describe the type of missingness mechanism. In particular, the missing data mechanism is ignorable (MAR) if $\lambda_1 = \lambda_2 = 0$; NI, otherwise. Following \citet{liu}, to estimate the parameter $\boldsymbol{\phi} = (\beta_0, \beta_1, \boldsymbol{\beta}^\top_2,\lambda_1, \lambda_2)^\top$ in the verification model (\ref{veri:model:1}), \citet{toduc2019vus} proposed to use a likelihood-based approach. Specifically, the authors assumed that the disease process  follows a multinomial logistic model for the whole sample, i.e., 
\begin{equation}
\rho_{k} = \frac{\exp\left(\eta^{(k)}_0 + \eta^{(k)}_1 T + \mathbf{A}^\top \boldsymbol{\eta}^{(k)}_2 \right)}{1 + \exp\left(\eta^{(1)}_0 + \eta^{(1)}_1 T + \mathbf{A}^\top \boldsymbol{\eta}^{(1)}_2 \right) + \exp\left(\eta^{(2)}_0 + \eta^{(2)}_1 T + \mathbf{A}^\top\boldsymbol{\eta}^{(2)}_2 \right)},
\label{dise:model}
\end{equation}
where $\rho_{k} = \Pr(D_k = 1| T, \mathbf{A}), k = 1,2$. Then, they jointly considered the verification model (\ref{veri:model:1}) and the disease model (\ref{dise:model}) to construct the (observed) log-likelihood function, $\ell(\boldsymbol{\phi}, \boldsymbol{\eta})$ say, where $\boldsymbol{\eta}= (\eta^{(1)}_0, \eta^{(1)}_1, \boldsymbol{\eta}^{(1)\top}_2, \eta^{(2)}_0, \eta^{(2)}_1, \boldsymbol{\eta}^{(2)\top}_2)^\top$  is the vector of parameters of the disease model. The authors proved that the joint model based on (\ref{veri:model:1}) and (\ref{dise:model}) is identifiable. Then, the maximum likelihood estimator (MLE) of $(\boldsymbol{\phi}^\top, \boldsymbol{\eta}^\top)^\top$ can be obtained by maximizing $\ell(\boldsymbol{\phi}, \boldsymbol{\eta})$ or solving the corresponding score equation. Based on imputation and re-weighting methods, four estimators for each disease status can be obtained: FI, with $\widehat{D}_{ki,\mathrm{FI}} = \hat{\rho}_{ki} = \widehat{\Pr}(D_k = 1|T_i, \mathbf{A}_i)$; MSI, with $\widehat{D}_{ki,\mathrm{MSI}} = V_i D_{ki} + (1 - V_i)\hat{\rho}_{k(0)i}$, where $\hat{\rho}_{k(0)i} = \widehat{\Pr}(D_{ki} = 1|T_i, \mathbf{A}_i, V_i = 0)$; IPW, with $\widehat{D}_{ki, \mathrm{IPW}} = V_i D_{ki} / \hat{\pi}_i$ and PDR, with $\widehat{D}_{ki,\mathrm{PDR}} = V_i D_{ki}/\hat{\pi}_i - \hat{\rho}_{k(0)i}(V_i - \hat{\pi}_i)/\hat{\pi}_i$, with $\hat{\pi}_i = \widehat{\Pr}(V_i = 1|T_i, \mathbf{A}_i, D_{1i}, D_{2i})$. Therefore, bias-corrected VUS estimators are derived by using quantities $\widehat{D}_{ki,*}$ in (\ref{nonp:vus}), i.e., 
\begin{equation}
\widehat{\mu}_{*} = \frac{\sum\limits_{i = 1}^{n}\sum\limits_{\ell = 1, \ell \ne i}^{n} \sum\limits_{\stackrel{r = 1}{r \ne \ell, r \ne i}}^{n}\I_{i\ell r}\widehat{D}_{1i,*}\widehat{D}_{2\ell,*}\widehat{D}_{3r,*}}{\sum\limits_{i = 1}^{n}\sum\limits_{\ell = 1, \ell \ne i}^{n} \sum\limits_{\stackrel{r = 1}{r \ne \ell, r \ne i}}^{n} \widehat{D}_{1i,*}\widehat{D}_{2\ell,*}\widehat{D}_{3r,*}}.
\label{bc:vus-old}
\end{equation} 
Here, the symbol $*$ stands for type of estimator, i.e., FI, MSI, IPW and PDR. 

Despite its usefulness, the above reviewed approach presents some limitations. Firstly, it is scarcely flexible, since identifiability is proven only for the specific parametric assumptions about the verification and  the disease processes. Secondly, because of NI missigness, it does not allow a statistical check of appropriateness of such assumptions (see also \citealp{riddles2016propensity}, \citealp{morikawa2017semiparametric} and \citealp{yu2018estimation}). Finally, as pointed out by the authors themselves, MLE of ignorable/nonignorable parameters $\lambda_1$ and $\lambda_2$ are typically highly biased, in cases of moderate or even high sample sizes.

In the present paper, we try to overcome such limitations. More precisely, here we consider  the mean score function as a tool to estimate the parameters in the verification model, and avoid identifiability problems by resorting to the introduction of instrumental variables \citep{wang2014instrumental, riddles2016propensity}. As we show in what follows, in order to solve the mean score equation derived by the chosen verification model, we preliminarily need to estimate the parameters of a model for the disease process, but its specification is required only for verified subjects. Thus, the new proposed approach essentially avoids the specification of a parametric regression model for the disease status of the whole sample, and presents the advantage that  it allows a  statistical evaluation, based  on the available sample data, of the adequacy of the required conditional (to $V=1$) disease model. Furthermore, as confirmed by our simulation experiments, the reduction in complexity of the adopted model translates into a greater stability of estimates of the parameters.

Based on the new approach, we provide four new VUS estimators (FI, MSI, IPW and PDR), which accommodate for nonignorable missingness of the disease status. We prove consistency and asymptotic normality of the proposed estimators, and also provide estimators of their variances.

The paper is organised as follows. In Section \ref{sec:2}, we present the proposed approach, and give the new bias--corrected VUS estimators, for which, in Section \ref{sec:3},  we discuss the asymptotic behaviour. In Section \ref{sec:simu}, we present  the results of the simulation study, and in Section \ref{sec:app} we give an illustrative application. Section \ref{sec:concl} contains some final remarks. Some technical details can be found in Appendix.

\section{The proposal}
\label{sec:2}
It is worth noting that MLE of $(\boldsymbol{\phi}^\top, \boldsymbol{\eta}^\top)^\top$,  on which estimators (\ref{bc:vus-old}) ultimately rely, could be also obtained as solution of the so-called mean score equation \citep{louis1982finding,kim2013statistical,riddles2016propensity}. Let $S(\boldsymbol{\phi})$ and $S(\boldsymbol{\eta})$ denote the score functions derived from the verification model (\ref{veri:model:1}) and from the disease model (\ref{dise:model}), respectively. For the joint model, the mean score function is $\bar{S}(\boldsymbol{\phi}, \boldsymbol{\eta}) =(\bar{S}(\boldsymbol{\phi}),\bar{S}(\boldsymbol{\eta}))^\top$ \citep{kim2013statistical}, with
\begin{eqnarray}
\bar{S}(\boldsymbol{\phi}) &=& 
\sum_{i = 1}^{n} \bigg[V_i h(\boldsymbol{\phi}; V_i, T_i, \mathbf{A}_{i}, D_{1i}, D_{2i}) + (1 - V_i)\E_0\left\{h(\boldsymbol{\phi}; V_i, T_i, \mathbf{A}_{i}, D_1, D_2)|T_i, \mathbf{A}_{i}\right\} \bigg]
\nonumber \\
\bar{S}(\boldsymbol{\eta}) &=& 
\sum_{i = 1}^{n} \bigg[V_i u(\boldsymbol{\eta}; T_i, \mathbf{A}_{i}, D_{1i}, D_{2i}) + (1 - V_i)\E_0\left\{u(\boldsymbol{\eta}; T_i, \mathbf{A}_{i}, D_1, D_2)|T_i, \mathbf{A}_{i}\right\} \bigg]
\label{mscore_rho},
\end{eqnarray}
where $h(\boldsymbol{\phi}; \cdot)$ and $u(\boldsymbol{\eta}; \cdot)$ denote the standard contributions to  $S(\boldsymbol{\phi})$ and $S(\boldsymbol{\eta})$, respectively. Terms $\E_0(\cdot|T_i, \mathbf{A}_{i}) = \E(\cdot|T_i, \mathbf{A}_{i}, V_i = 0)$, represent the contribution to the scores, arising from unverified subjects.

Hereafter, we consider an alternative estimation procedure which solves the mean score equation derived from the verification model only, and uses instrumental variables. To this aim,
we assume that the vector of covariates $\mathbf{A}$ can be decomposed as $\mathbf{A} = \left(\mathbf{A}^\top_1, \mathbf{A}^\top_2\right)^\top$ such that the verification status $V$ does not depend on $\mathbf{A}_2$ given $T$, $\mathbf{A}_1$, $D_1$ and $D_2$.  Variable $\mathbf{A}_2$ is known as instrumental variable and helps to make the verification model identifiable (see \citealp{wang2014instrumental} and \citealp{riddles2016propensity}) without further assumptions on the disease model. By using the instrumental variable $\mathbf{A}_2$, we can reduce the verification model (\ref{veri:model:1}) to
\begin{eqnarray}
\pi &=& \Pr(V = 1|T, \mathbf{A}, D_1, D_2) = \Pr(V = 1|T, \mathbf{A}_1, D_1, D_2), \nonumber\\
&=& \frac{\exp\left(\beta_0 + \beta_1 T + \boldsymbol{\beta}_{2,1} ^\top \mathbf{A}_1 + \lambda_1 D_1 + \lambda_2 D_2 \right)}{1 + \exp\left(\beta_0 + \beta_1 T + \boldsymbol{\beta}_{2,1} ^\top \mathbf{A}_1 + \lambda_1 D_1 + \lambda_2 D_2 \right)} = \pi(T, \mathbf{A}_1, D_1, D_2; \boldsymbol{\gamma})
\label{veri:model:2}
\end{eqnarray}
where $\boldsymbol{\gamma} = \left(\beta_0, \beta_1, \boldsymbol{\beta}^\top_{2,1}, \lambda_1, \lambda_2\right)^\top$. Then, MLE estimate of $\boldsymbol{\gamma}$ can be obtained by solving the mean score equation $\bar{S}(\boldsymbol{\gamma}) = \boldsymbol{0}$ where
\begin{eqnarray}
\bar{S}(\boldsymbol{\gamma}) &=&  
\sum_{i = 1}^{n} \bigg[V_i h(\boldsymbol{\gamma}; V_i, T_i, \mathbf{A}_{1i}, D_{1i}, D_{2i}) + (1 - V_i)\E_0\left\{h(\boldsymbol{\gamma}; V_i, T_i, \mathbf{A}_{1i}, D_1, D_2)|T_i, \mathbf{A}_{i}\right\} \bigg]
\label{score:1}
\end{eqnarray}
with
\[
h(\boldsymbol{\gamma}; V_i, T_i, \mathbf{A}_{1i}, D_{1i}, D_{2i}) =
Z_i \left\{V_i - \pi(T_i, \mathbf{A}_{1i}, D_{1i}, D_{2i};\boldsymbol{\gamma}) \right\}
\]
for $Z_i= (1, T_i, \mathbf{A}_{1i}^\top, D_{1i}, D_{2i})^\top$ and
\begin{eqnarray}
\lefteqn{\E_0\left\{h(\boldsymbol{\gamma}; V_i, T_i, \mathbf{A}_{1i}, D_{1i}, D_{2i})|T_i, \mathbf{A}_i \right\}} \nonumber \\
&=& \E\left\{h(\boldsymbol{\gamma}; 0, T_i, \mathbf{A}_{1i}, D_{1i}, D_{2i})|T_i, \mathbf{A}_i, V_i = 0 \right\} \nonumber \\
&=& \Pr(D_{1i} = 1| T_i, \mathbf{A}_i, V_i = 0) h(\boldsymbol{\gamma}; 0, T_i, \mathbf{A}_{1i}, 1, 0) + \Pr(D_{2i} = 1| T_i, \mathbf{A}_i, V_i = 0)h(\boldsymbol{\gamma}; 0, T_i, \mathbf{A}_{1i}, 0, 1) \nonumber \\
&& + \: \Pr(D_{3i} = 1| T_i, \mathbf{A}_i, V_i = 0) h(\boldsymbol{\gamma}; 0, T_i, \mathbf{A}_{1i}, 0, 0) \label{expr:E0}.
\end{eqnarray}
Let $\rho_{k(v)}$ denote $\Pr(D_k = 1|V = v, T, \mathbf{A})$ for $k = 1, 2, 3$ and $v = 0, 1$. By an application of Bayes' rule (see  Appendix \ref{app:A}), it can be shown that
\begin{eqnarray}
\rho_{1(0)}(\lambda_1, \lambda_2) &=& \frac{\rho_{1(1)} e^{\lambda_2}}{\rho_{1(1)} e^{\lambda_2} + \rho_{2(1)} e^{\lambda_1} + \rho_{3(1)} e^{\lambda_1 + \lambda_2}}, \nonumber \\ 
\rho_{2(0)}(\lambda_1, \lambda_2) &=& \frac{\rho_{2(1)} e^{\lambda_1}}{\rho_{1(1)} e^{\lambda_2} + \rho_{2(1)} e^{\lambda_1} + \rho_{3(1)} e^{\lambda_1 + \lambda_2}}, \label{expr:rho_0} \\
\rho_{3(0)}(\lambda_1, \lambda_2) &=& \frac{\rho_{3(1)} e^{\lambda_1 + \lambda_2}}{\rho_{1(1)} e^{\lambda_2} + \rho_{2(1)} e^{\lambda_1} + \rho_{3(1)} e^{\lambda_1 + \lambda_2}}. \nonumber
\end{eqnarray}
Hence, the mean score function (\ref{score:1}) depends on the disease probabilities for verified units only. If (just to keep the parallel with the proposal in \citet{toduc2019vus}) we use a multinomial logistic regression model for the conditional disease probabilities $\rho_{k(1)}=\Pr(D_k = 1| V = 1, T, \mathbf{A})$, i.e.,
\begin{eqnarray}
\rho_{k(1)} &\equiv& \frac{\exp\left(\eta^{(k)}_0 + \eta^{(k)}_1 T + \mathbf{A}^\top \boldsymbol{\eta}^{(k)}_2 \right)}{1 + \exp\left(\eta^{(1)}_0 + \eta^{(1)}_1 T + \mathbf{A}^\top \boldsymbol{\eta}^{(1)}_2 \right) + \exp\left(\eta^{(2)}_0 + \eta^{(2)}_1 T + \mathbf{A}^\top\boldsymbol{\eta}^{(2)}_2 \right)},  \nonumber\\
&=& \rho_{k(1)}(\boldsymbol{\eta})
\label{dise:model:veri}
\end{eqnarray}
$k = 1, 2$, the functions $\rho_{k(0)}(\lambda_1, \lambda_2)$ depend also on $\boldsymbol{\eta}$,
and the expressions in (\ref{expr:rho_0}) become 
\begin{eqnarray}
\rho_{1(0)}(\boldsymbol{\eta}, \lambda_1, \lambda_2) &=& \frac{\rho_{1(1)}(\boldsymbol{\eta}) e^{\lambda_2}}{\rho_{1(1)}(\boldsymbol{\eta}) e^{\lambda_2} + \rho_{2(1)}(\boldsymbol{\eta}) e^{\lambda_1} + \rho_{3(1)}(\boldsymbol{\eta}) e^{\lambda_1 + \lambda_2}}, \nonumber \\ 
\rho_{2(0)}(\boldsymbol{\eta}, \lambda_1, \lambda_2) &=& \frac{\rho_{2(1)}(\boldsymbol{\eta}) e^{\lambda_1}}{\rho_{1(1)}(\boldsymbol{\eta}) e^{\lambda_2} + \rho_{2(1)}(\boldsymbol{\eta}) e^{\lambda_1} + \rho_{3(1)}(\boldsymbol{\eta}) e^{\lambda_1 + \lambda_2}}, \label{expr:rho_0:eta} \\
\rho_{3(0)}(\boldsymbol{\eta}, \lambda_1, \lambda_2) &=& \frac{\rho_{3(1)}(\boldsymbol{\eta}) e^{\lambda_1 + \lambda_2}}{\rho_{1(1)}(\boldsymbol{\eta}) e^{\lambda_2} + \rho_{2(1)}(\boldsymbol{\eta}) e^{\lambda_1} + \rho_{3(1)}(\boldsymbol{\eta}) e^{\lambda_1 + \lambda_2}}. \nonumber
\end{eqnarray}
By using (\ref{expr:E0}) and (\ref{expr:rho_0:eta}), the mean score function (\ref{score:1}) can be rewritten as
\begin{eqnarray}
\bar{S}(\boldsymbol{\gamma}; \boldsymbol{\eta}) &=& \sum_{i = 1}^{n} s_i(\boldsymbol{\gamma}; \boldsymbol{\eta}) \nonumber \\
&=& \sum_{i = 1}^{n} \bigg[V_i h(\boldsymbol{\gamma}; 1, T_i, \mathbf{A}_{1i}, D_{1i}, D_{2i}) + (1 - V_i)\big\{\rho_{1(0)i}(\boldsymbol{\eta}, \lambda_1, \lambda_2) h(\boldsymbol{\gamma}; 0, T_i, \mathbf{A}_{1i}, 1, 0) \nonumber \\
&& + \: \rho_{2(0)i}(\boldsymbol{\eta}, \lambda_1, \lambda_2) h(\boldsymbol{\gamma}; 0, T_i, \mathbf{A}_{1i}, 0, 1) + \rho_{3(0)i}(\boldsymbol{\eta}, \lambda_1, \lambda_2) h(\boldsymbol{\gamma}; 0, T_i, \mathbf{A}_{1i}, 0, 0)\big\} \bigg].
\label{score:2:gen}
\end{eqnarray}
Then, once an estimate $\widehat{\boldsymbol{\eta}}$ is available, we can get an empirical form of the mean score function (\ref{score:1}), say $\bar{S}(\boldsymbol{\gamma}; \widehat{\boldsymbol{\eta}})$. The empirical mean score equation $\bar{S}(\boldsymbol{\gamma}; \widehat{\boldsymbol{\eta}}) = \boldsymbol{0}$ can be solved to obtain an estimate $\widehat{\boldsymbol{\gamma}}$ of $\boldsymbol{\gamma}$. 

Observe that $\widehat{\boldsymbol{\gamma}}$ would be the MLE estimate only if $\widehat{\boldsymbol{\eta}}$ was solution of equation (\ref{mscore_rho}). It should be noted, however, that in our simplified approach, $\widehat{\boldsymbol{\eta}}$ is the solution of a  score equation that depends on the parametric model specification for the conditional (to $V=1$) disease process. The verification mechanism may be modelled by means of a generalised linear regression model for binary data with a fixed link function (e.g., logistic, probit, log-log or complementary log-log). Whichever choice is taken, identifiability of the model holds thanks to the use of instrumental variables. In what follows, for the purposes of presentation, we  adopt a standard logistic model for the verification process, and a multinomial logistic model for the conditional (to verification) disease process, a solution that also allow to keep the parallel with the approach in \citet{toduc2019vus}. 

Let $U(\boldsymbol{\eta}) = \sum\limits_{i = 1}^{n} V_iu_i(\boldsymbol{\eta})$, with $u_i=u(\boldsymbol{\eta}; T_i, \mathbf{A}_{i}, D_{1i}, D_{2i})$, denote the estimating function for $\boldsymbol{\eta}$ based on the verified units in the sample. Since our empirical mean score equation $\bar{S}(\boldsymbol{\gamma}; \widehat{\boldsymbol{\eta}}) = \boldsymbol{0}$ is a specific case of a more general one discussed by \citet{riddles2016propensity}, it follows that, under certain regularity conditions given in Appendix \ref{app:B}, $\widehat{\boldsymbol{\gamma}}$ and $\widehat{\boldsymbol{\eta}}$ are consistent and asymptotically normal estimators (see Theorem 1 in \citealp{riddles2016propensity}).

In order to obtain four bias-corrected VUS estimators (FI, MSI, IPW and PDR), we follow the arguments in \citet{toduc2019vus}, starting by the fact that
\[
\E\left[V_i \rho_{k(1)i}(\boldsymbol{\eta}) + (1 - V_i)\rho_{k(0)i}(\boldsymbol{\eta}, \lambda_1, \lambda_2) | T_i, \mathbf{A}_i \right] = \Pr(D_k = 1|T_i, \mathbf{A}_i)
\]
and considering the quantities 
\[
D_{ki,\mathrm{FI}}(\boldsymbol{\eta}, \lambda_1, \lambda_2) = V_i \rho_{k(1)i}(\boldsymbol{\eta}) + (1- V_i)\rho_{k(0)i}(\boldsymbol{\eta}, \lambda_1, \lambda_2),
\]
\[
D_{ki,\mathrm{MSI}}(\boldsymbol{\eta}, \lambda_1, \lambda_2) = V_i D_{ki} + (1 - V_i)\rho_{k(0)i}(\boldsymbol{\eta}, \lambda_1, \lambda_2),
\]
\[
D_{ki, \mathrm{IPW}}(\boldsymbol{\gamma}) = \frac{V_i D_{ki}}{\pi_i(\boldsymbol{\gamma})}
\ \ \  {\rm and} \ \ \
D_{ki,\mathrm{PDR}}(\boldsymbol{\eta}, \boldsymbol{\gamma}) = \frac{V_i D_{ki}}{\pi_i(\boldsymbol{\gamma})} - \rho_{k(0)i}(\boldsymbol{\eta}, \lambda_1, \lambda_2)\frac{V_i - \pi_i(\boldsymbol{\gamma})}{\pi_i(\boldsymbol{\gamma})}.
\] 
in which we stressed the fact that the probabilities of selection for verification $\pi_i$ depend on 
$\boldsymbol{\gamma}$.
Plugging in the above expressions the estimates $\widehat{\boldsymbol{\eta}}$ and $\widehat{\boldsymbol{\gamma}}$, yields  the estimates $\widehat{D}_{ki,\mathrm{FI}}$, $\widehat{D}_{ki,\mathrm{MSI}}$, $\widehat{D}_{ki, \mathrm{IPW}}$ and $\widehat{D}_{ki, \mathrm{PDR}}$.
Then, the bias-corrected VUS estimators are obtained by replacing each disease status in (\ref{nonp:vus})  with the new estimates $\widehat{D}_{ki,*}$, where the symbol $*$ indicates, again, the kind of estimator, i.e., FI, MSI, IPW and PDR.

Note that, the parameters $\boldsymbol{\gamma}$ of the verification model (\ref{veri:model:2}) and $\boldsymbol{\eta}$ of the conditional (to $V=1$) disease model (\ref{dise:model:veri}) are estimated separately (see the estimation process in Appendix \ref{app:C}). However, the accuracy of  the estimates $\widehat{\pi}_i$ in the verification model (\ref{veri:model:2}) may suffer from possible fitting's problems in the disease model (\ref{dise:model:veri}), because the empirical mean score equation uses $\widehat{\boldsymbol{\eta}}$. On the other hand, the estimates $\widehat{\rho}_{k(0)i}$  also may be affected by possible poor accuracy in the estimation of ${\lambda}_1$, ${\lambda}_2$ or ${\boldsymbol{\eta}}$. Ultimately, the fact that the proposed method is a parametric method has as a consequence that it could produce VUS estimates even strongly influenced by incorrect specifications of the involved models. 

\section{Asymptotic behaviour}
\label{sec:3}
As for the estimators discussed in \citet{toduc2019vus},  the new proposed VUS estimators can be found as solutions of appropriate estimating equations, solved along with the score equations $U(\boldsymbol{\eta}) = \boldsymbol{0}$ and the empirical mean score equation $\bar{S}(\boldsymbol{\gamma}; \widehat{\boldsymbol{\eta}}) = \boldsymbol{0}$. The estimating functions for FI, MSI, IPW and PDR estimators  have generic term (corresponding to a generic triplet of sample units), respectively,  
\begin{eqnarray}
G_{i\ell r, \text{FI}}(\mu, \boldsymbol{\eta},\boldsymbol{\gamma})
&=& D_{1i,\text{FI}}(\boldsymbol{\eta}, \lambda_1, \lambda_2) D_{2\ell, \text{FI}}(\boldsymbol{\eta}, \lambda_1, \lambda_2) D_{3r, \text{FI}}(\boldsymbol{\eta}, \lambda_1, \lambda_2) \left(I_{i\ell r} - \mu\right), \nonumber \\
G_{i\ell r, \text{MSI}}(\mu, \boldsymbol{\eta}, \boldsymbol{\gamma})
&=& D_{1i,\text{MSI}}(\boldsymbol{\eta}, \lambda_1, \lambda_2) D_{2\ell, \text{MSI}}(\boldsymbol{\eta}, \lambda_1, \lambda_2) D_{3r, \text{MSI}}(\boldsymbol{\eta}, \lambda_1, \lambda_2) \left(I_{i\ell r} - \mu\right), \nonumber \\
G_{i\ell r, \text{IPW}}(\mu, \boldsymbol{\eta}, \boldsymbol{\gamma})
&=& D_{1i, \mathrm{IPW}}(\boldsymbol{\gamma}) D_{2\ell, \mathrm{IPW}}(\boldsymbol{\gamma}) D_{3r, \mathrm{IPW}}(\boldsymbol{\gamma}) \left(I_{i\ell r} - \mu\right), \nonumber \\
G_{i\ell r,\text{PDR}}(\mu, \boldsymbol{\eta}, \boldsymbol{\gamma})
&=& D_{1i, \text{PDR}}(\boldsymbol{\eta}, \boldsymbol{\gamma}) D_{2\ell, \text{PDR}}(\boldsymbol{\eta}, \boldsymbol{\gamma}) D_{3r, \text{PDR}}(\boldsymbol{\eta}, \boldsymbol{\gamma}) \left(I_{i\ell r} - \mu\right), \nonumber
\end{eqnarray}
for $i \ne \ell$, $\ell \ne r$ and $r \ne i$.
In the following, we will use the general notation $G_{i\ell r,*}(\mu, \boldsymbol{\eta}, \boldsymbol{\gamma})$, where  $*$ stands for FI, MSI, IPW and PDR, and denote by 
$\widetilde{\mu}_*$ the corresponding new VUS estimator. Moreover, let $\boldsymbol{\eta}_0$, $\boldsymbol{\gamma}_0$ and $\mu_0$ be the true parameter values.  We assume that
\begin{enumerate} [(C1)]
	\item the U--process, $\sqrt{n}\left\{G_{*}(\mu, \boldsymbol{\eta}, \boldsymbol{\gamma}) - e_*(\mu, \boldsymbol{\eta}, \boldsymbol{\gamma})\right\}$, is stochastically equicontinuous, where
	\begin{eqnarray}
	G_{*}(\mu, \boldsymbol{\eta}, \boldsymbol{\gamma})
	&=& \frac{1}{6n(n-1)(n-2)} \sum\limits_{i = 1}^{n}\sum\limits_{\ell = 1, \ell \ne i}^{n}\sum\limits_{\stackrel{r = 1}{r \ne \ell, r \ne i}}^{n} \bigg\{
	G_{i\ell r,*}(\mu, \boldsymbol{\eta}, \boldsymbol{\gamma}) + G_{ir \ell,*}(\mu, \boldsymbol{\eta}, \boldsymbol{\gamma}) \nonumber\\
	&& + \: G_{\ell ir,*}(\mu, \boldsymbol{\eta}, \boldsymbol{\gamma}) + G_{\ell r i,*}(\mu, \boldsymbol{\eta}, \boldsymbol{\gamma}) + G_{r i\ell ,*}(\mu, \boldsymbol{\eta}, \boldsymbol{\gamma}) + G_{r \ell i,*}(\mu, \boldsymbol{\eta}, \boldsymbol{\gamma}) \bigg\} \nonumber 
	\end{eqnarray}
	and
	\begin{eqnarray}
	e_*(\mu, \boldsymbol{\eta}, \boldsymbol{\gamma}) &=& \frac{1}{6} \E \bigg\{
	G_{i\ell  r,*}(\mu, \boldsymbol{\eta}, \boldsymbol{\gamma}) + G_{ir \ell ,*}(\mu, \boldsymbol{\eta}, \boldsymbol{\gamma}) + G_{\ell ir,*}(\mu, \boldsymbol{\eta}, \boldsymbol{\gamma}) + G_{\ell r i,*}(\mu, \boldsymbol{\eta}, \boldsymbol{\gamma}) \nonumber\\
	&& + \: G_{r i\ell ,*}(\mu, \boldsymbol{\eta}, \boldsymbol{\gamma}) + G_{r \ell i,*}(\mu, \boldsymbol{\eta}, \boldsymbol{\gamma}) \bigg\} \nonumber;
	\end{eqnarray}
	\item $e_*(\mu, \boldsymbol{\eta}, \boldsymbol{\gamma})$ is differentiable in $(\mu, \boldsymbol{\eta}, \boldsymbol{\gamma})$, and
	$\dfrac{\partial e_*(\mu, \boldsymbol{\eta}_0, \boldsymbol{\gamma}_0)}{\partial \mu}\Bigg|_{\mu = \mu_0} \ne 0$;
	\item $G_{*}(\mu,\boldsymbol{\eta}, \boldsymbol{\gamma})$ and $\frac{\partial G_{*}(\mu, \boldsymbol{\eta}, \boldsymbol{\gamma})}{\partial (\boldsymbol{\eta}, \boldsymbol{\gamma})^\top}$ converge uniformly (in probability) to $e_*(\mu, \boldsymbol{\eta}, \boldsymbol{\gamma})$ and $\frac{\partial e_*(\mu, \boldsymbol{\eta}, \boldsymbol{\gamma})}{\partial (\boldsymbol{\eta}, \boldsymbol{\gamma})^\top}$, respectively.
\end{enumerate}

Let $\mathbf{O}_i = (T_i, \mathbf{A}^\top_{i}, D_{1i}, D_{2i}, D_{3i},V_i)^\top$, $\mathcal{I}_{\boldsymbol{\eta}}(\boldsymbol{\eta}) = \E \left\{- \frac{\partial u_i(\boldsymbol{\eta})}{\partial \boldsymbol{\eta}^\top} \right\}$, $\mathcal{I}_{s,\boldsymbol{\gamma}}(\boldsymbol{\eta}, \boldsymbol{\gamma}) = \E \left\{- \frac{\partial s_i(\boldsymbol{\gamma}; \boldsymbol{\eta})}{\partial \boldsymbol{\gamma}^\top} \right\}$ and $\mathcal{I}_{s,\boldsymbol{\eta}}(\boldsymbol{\eta}, \boldsymbol{\gamma}) = \E \left\{- \frac{\partial s_i(\boldsymbol{\gamma}; \boldsymbol{\eta})}{\partial \boldsymbol{\eta}^\top} \right\}$. 

\begin{theorem}
	Suppose that conditions (C1)--(C3), and (D1)--(D9) in Appendix B, hold. Then, under the verification model (\ref{veri:model:2}) and the (conditional) disease model (\ref{dise:model:veri}), $\widetilde{\mu}_*$ is consistent and asymptotically normal, i.e.,
	\begin{equation}
	\sqrt{n}\left(\widetilde{\mu}_{*} - \mu_0 \right) \stackrel{d}{\longrightarrow} \mathcal{N}(0, \sigma_*^2), \nonumber
	\end{equation}
	where \ $\sigma_*^2 = \dfrac{\Var \left\{\Lambda_{i,*} + Q_{i,*}\right\}}{\left[\Pr(D_1 = 1) \Pr(D_2 = 1) \Pr(D_3 = 1)\right]^2}$ \ and
	\begin{eqnarray}
	\Lambda_{i,*} &=& \frac{1}{2} \E \bigg\{ G_{i\ell r,*}(\mu_0,\boldsymbol{\eta}_0, \boldsymbol{\gamma}_0) + G_{ir \ell ,*}(\mu_0, \boldsymbol{\eta}_0, \boldsymbol{\gamma}_0) + G_{\ell ir,*}(\mu_0, \boldsymbol{\eta}_0, \boldsymbol{\gamma}_0) \nonumber\\
	&& + \: G_{\ell r i,*}(\mu_0, \boldsymbol{\eta}_0, \boldsymbol{\gamma}_0) + G_{r i\ell ,*}(\mu_0, \boldsymbol{\eta}_0, \boldsymbol{\gamma}_0) + G_{r \ell i,*}(\mu_0, \boldsymbol{\eta}_0, \boldsymbol{\gamma}_0) \big|\mathbf{O}_i \bigg\} \nonumber, \\
	Q_{i,*} &=& \frac{\partial e_*(\mu_0, \boldsymbol{\eta}, \boldsymbol{\gamma})}{\partial (\boldsymbol{\eta}, \boldsymbol{\gamma})^\top}\bigg|_{(\boldsymbol{\eta}, \boldsymbol{\gamma})^\top = (\boldsymbol{\eta}_0, \boldsymbol{\gamma}_0)^\top} \begin{pmatrix}
	\mathcal{I}_{\boldsymbol{\eta}}(\boldsymbol{\eta}_0) & \boldsymbol{0} \\
	\mathcal{I}_{s,\boldsymbol{\eta}}(\boldsymbol{\eta}_0, \boldsymbol{\gamma}_0) & \mathcal{I}_{s,\boldsymbol{\gamma}}(\boldsymbol{\eta}_0, \boldsymbol{\gamma}_0)
	\end{pmatrix}^{-1}  \begin{pmatrix}
	V_i u_i(\boldsymbol{\eta}_0) \\ s_i(\boldsymbol{\gamma}_0; \boldsymbol{\eta}_0)
	\end{pmatrix}. \nonumber
	\end{eqnarray}
\end{theorem}
\begin{proof}(sketch).
	Under the verification model (\ref{veri:model:2}) and the disease model (\ref{dise:model:veri}), we can show that quantities $G_{i\ell r,*}(\mu_0, \boldsymbol{\eta}_0, \boldsymbol{\gamma}_0)$ have 0 expectation. Then, the consistency of $\widetilde{\mu}_{*}$ can be proven by arguments similar to those used  in Theorem 1 of \citet{toduc2019vus}. As for the asymptotic normality, we follow the steps in the proof of Theorem 2 in \citet{toduc2019vus}. In particular, using also equation (\ref{asy:gamma:eta}), since $e_*(\mu_0, \boldsymbol{\eta}_0, \boldsymbol{\gamma}_0) = 0$, we can write
	\begin{eqnarray}
	0 &=& \frac{\partial e_*(\mu, \boldsymbol{\eta}_0, \boldsymbol{\gamma}_0)}{\partial \mu}\bigg|_{\mu = \mu_0}\sqrt{n}(\widetilde{\mu}_* - \mu_0) + \sqrt{n}G_{*}(\mu_0, \boldsymbol{\eta}_0, \boldsymbol{\gamma}_0)  \nonumber \\ 
	&& + \: \frac{\partial e_*(\mu_0, \boldsymbol{\eta}, \boldsymbol{\gamma})}{\partial (\boldsymbol{\eta}, \boldsymbol{\gamma})^\top}\bigg|_{(\boldsymbol{\eta}, \boldsymbol{\gamma})^\top = (\boldsymbol{\eta}_0, \boldsymbol{\gamma}_0)^\top} \sqrt{n}\begin{pmatrix}
	\widehat{\boldsymbol{\eta}} - \boldsymbol{\eta}_0 \\ 
	\widehat{\boldsymbol{\gamma}} - \boldsymbol{\gamma}_0
	\end{pmatrix} + o_p(1).
	\label{expan:Taylor}
	\end{eqnarray}
	An application of a standard result about the limit distribution of U-statistics \citep{van}, yields to
	\[
	\sqrt{n}G_{*}(\mu_0, \boldsymbol{\eta}_0, \boldsymbol{\gamma}_0) \stackrel{d}{\longrightarrow} \frac{1}{\sqrt{n}} \sum_{i = 1}^{n} \Lambda_{i,*} \ .
	\]
	This, together with the fact that
	\[
	\frac{\partial e_*(\mu, \boldsymbol{\eta}_0, \boldsymbol{\gamma}_0)}{\partial \mu}\bigg|_{\mu = \mu_0} = -\Pr(D_1 = 1) \Pr(D_2 = 1) \Pr(D_3 = 1),
	\]
	and equation (\ref{asy:gamma:eta}), show that (\ref{expan:Taylor}) is equivalent to
	\begin{eqnarray}
	\lefteqn{\Pr(D_1 = 1) \Pr(D_2 = 1) \Pr(D_3 = 1) \sqrt{n}(\widetilde{\mu}_* - \mu_0)} \nonumber \\
	&=& \frac{\partial e_*(\mu_0, \boldsymbol{\eta}, \boldsymbol{\gamma})}{\partial (\boldsymbol{\eta}, \boldsymbol{\gamma})^\top}\bigg|_{(\boldsymbol{\eta}, \boldsymbol{\gamma})^\top = (\boldsymbol{\eta}_0, \boldsymbol{\gamma}_0)^\top} \begin{pmatrix}
	\mathcal{I}_{\boldsymbol{\eta}}(\boldsymbol{\eta}_0) & \boldsymbol{0} \\
	\mathcal{I}_{s,\boldsymbol{\eta}}(\boldsymbol{\eta}_0, \boldsymbol{\gamma}_0) & \mathcal{I}_{s,\boldsymbol{\gamma}}(\boldsymbol{\eta}_0, \boldsymbol{\gamma}_0)
	\end{pmatrix}^{-1} \frac{1}{\sqrt{n}}\sum_{i = 1}^{n} \begin{pmatrix}
	V_iu_i(\boldsymbol{\eta}_0) \\ s_i(\boldsymbol{\gamma}_0; \boldsymbol{\eta}_0)
	\end{pmatrix} \nonumber\\
	&& + \: \frac{1}{\sqrt{n}} \sum_{i = 1}^{n} \Lambda_{i,*} + o_p(1). \nonumber
	\end{eqnarray}
	Then, we have that
	\[
	\sqrt{n}(\widetilde{\mu}_* - \mu_0) = \frac{1}{\sqrt{n}}\sum_{i = 1}^{n} \frac{\Lambda_{i,*} + Q_{i,*}}{\Pr(D_1 = 1) \Pr(D_2 = 1) \Pr(D_3 = 1)} + o_p(1),
	\]
	and the asymptotic normality of $\widetilde{\mu}_*$ follows by the Central Limit Theorem.
\end{proof}

The asymptotic variances $\sigma_*^2$ can be consistently estimated by
\begin{eqnarray}
\hat{\sigma}_*^2 &=& \frac{1}{n-1}\sum_{i = 1}^{n}\left\{\left(\hat{\Lambda}_{i,*} + \widehat{Q}_{i,*}\right) - \left(\overline{\Lambda}_* + \overline{Q}_*\right) \right\}^2 \nonumber\\
&& \times \: \left\{\widehat{\Pr}_*(D_1 = 1) \widehat{\Pr}_*(D_2 = 1) \widehat{\Pr}_*(D_3 = 1)\right\}^{-2},
\label{var_est}
\end{eqnarray}
where
\begin{eqnarray}
\hat{\Lambda}_{i,*} &=& \frac{1}{(n-1)(n-2)} \sum_{\stackrel{\ell = 1}{\ell \ne i}}^{n}\sum_{\stackrel{r = 1}{r \ne \ell, r \ne i}}^{n} \bigg\{G_{i \ell r,*}(\widehat{\mu}_*,\widehat{\boldsymbol{\eta}}, \widehat{\boldsymbol{\gamma}}) + G_{\ell i r,*}(\widehat{\mu}_*,\widehat{\boldsymbol{\eta}}, \widehat{\boldsymbol{\gamma}}) \nonumber\\
&& + \: G_{r \ell i,*}(\widehat{\mu}_*,\widehat{\boldsymbol{\eta}}, \widehat{\boldsymbol{\gamma}})\bigg\} \nonumber \\
\widehat{Q}_{i,*} &=& \frac{1}{(n-1)(n-2)}\sum_{i = 1}^{n}\sum_{\stackrel{\ell  = i}{\ell  \ne i}}^{n}\sum_{\stackrel{r = 1}{r \ne \ell , r \ne i}}^{n} \frac{\partial G_{i\ell r,*}(\widehat{\mu}_{*},\boldsymbol{\eta}, \boldsymbol{\gamma})}{\partial (\boldsymbol{\eta}, \boldsymbol{\gamma})^\top}\bigg|_{(\boldsymbol{\eta}, \boldsymbol{\gamma})^\top = (\widehat{\boldsymbol{\eta}}, \widehat{\boldsymbol{\gamma}})^\top} \nonumber \\
&& \times \: \begin{pmatrix}
\sum\limits_{i = 1}^{n}\frac{\partial}{\partial \boldsymbol{\eta}^\top} V_i u_i(\widehat{\boldsymbol{\eta}}) & \boldsymbol{0} \\
\sum\limits_{i = 1}^{n} \frac{\partial}{\partial \boldsymbol{\eta}^\top} s_i(\widehat{\boldsymbol{\gamma}}; \widehat{\boldsymbol{\eta}}) & \sum\limits_{i = 1}^{n} \frac{\partial}{\partial \boldsymbol{\gamma}^\top} s_i(\widehat{\boldsymbol{\gamma}}; \widehat{\boldsymbol{\eta}})
\end{pmatrix}^{-1} \begin{pmatrix}
V_i u_i(\widehat{\boldsymbol{\eta}}) \\ s_i(\widehat{\boldsymbol{\gamma}}; \widehat{\boldsymbol{\eta}})
\end{pmatrix} \nonumber
\end{eqnarray}
and
\[
\widehat{\Pr}_*(D_k = 1) = \frac{1}{n}\sum_{i = 1}^{n} \widehat{D}_{ki,*}, \qquad  \overline{\Lambda}_* = \frac{1}{n}\sum_{i = 1}^{n}\hat{\Lambda}_{i,*}, \qquad \overline{Q}_* = \frac{1}{n}\sum_{i = 1}^{n}\widehat{Q}_{i,*}.
\]

Clearly, alternative approaches to estimate the variance of the proposed VUS estimators are possible. For instance, one may resort to the classical bootstrap method.

\section{Simulation study}
\label{sec:simu}
In order to compare the  proposed mean score equation-based approach with the fully likelihood-based method in \citep{toduc2019vus} and to evaluate the performance of the new bias-corrected VUS estimators, we carry out  several simulation experiments. The number of replications is 1000 in each experiment, and the considered sample sizes are $150, 250, 500$ and $1000$. 

\subsection{Simulation setup}
\label{sec:simu:setup}
With regard to data generation, we consider the following six different scenarios: 
\begin{enumerate}[I.]
	\item For each unit, the test result $T$ and a covariate $A$ are generated by the  model
	\[
	(T, A) \sim \mathcal{N}\left( \begin{pmatrix} 3.7 \\ 1.85 \end{pmatrix}, 
	\begin{pmatrix}
	3.71 & 1.36 \\
	1.36 & 3.13
	\end{pmatrix} \right).
	\]
	The disease status $\mathcal{D} = (D_1, D_2, D_3)^\top$ is generated by the following multinomial logistic model:
	\begin{eqnarray}
	\log\left( \frac{\Pr(D_1 = 1| T, A)}{\Pr(D_3 = 1| T, A)} \right) &=& 15 - 3.3T - 0.7A, \nonumber \\
	\log\left( \frac{\Pr(D_2 = 1| T, A)}{\Pr(D_3 = 1| T, A)} \right) &=& 9.5 - 1.7T - 0.3A. \nonumber
	\end{eqnarray}
	The verification status $V$ is simulated by the following logistic model
	\[
	\Pr(V = 1| T, A, D_1, D_2) = \frac{\exp(2 + 0.5T - 1.2A - 2D_1 - D_2)}{1 + \exp(2 + 0.5T - 1.2A - 2D_1 - D_2)}.
	\]
	Under this setting, the verification rate is roughly 0.57 and the true VUS is 0.791. This setting is the same as Scenario I in \citet{toduc2019vus}.
	\item For each unit, the disease status $\mathcal{D}= (D_{1}, D_{2}, D_{3})^\top$ is generated from a multinomial distribution with $\theta_1= 0.7$, $\theta_2= 0.2$  and $\theta_3=0.1$ (recall that $\theta_k=\Pr(D_{k}=1),$ $k=1,2,3$). The test results $T$ and a covariate $A$ are generated as $T|D_{k} = 1 \sim \mathcal{N}(k - 1, 0.5^2)$ and $A|D_{k} = 1 \sim \mathcal{N}(0.5(k - 1), 0.5^2)$, while the verification status $V$ is simulated through the model
	\[
	\Pr(V = 1| T, D_1, D_2) = \frac{\exp(1 + T - 2D_1 - D_2)}{1 + \exp(1 + T - 2D_1 - D_2)}.
	\]
	The verification rate is roughly 0.44 and the true VUS is 0.843.
	\item The disease status $\mathcal{D}$ is generated as in scenario I. The test results $T$ and a covariate $A_{1}$ are generated as $T|D_{k} = 1 \sim \mathcal{N}(0.4(k - 1), 0.5^2)$ and $A_{1}|D_{k} = 1 \sim \mathcal{N}(0.5(k - 1), 0.5^2)$, for $k = 1,2,3$. A second covariate $A_{2}$ is simulated as $A_{2}|D_{1} = 1 \sim \mathcal{U}(-2,-1)$, $A_{2}|D_{2} = 1 \sim \mathcal{U}(-1,1)$ and $A_{2}|D_{3} = 1 \sim \mathcal{U}(1,2)$. The true VUS is 0.457. The verification status $V$ is simulated by a probit model
	\[
	\Pr(V = 1| T, A_1, D_1, D_2) = \Phi(1.5 + T - 0.5A_1 - 2D_1 - D_2),
	\]
	where $\Phi(\cdot)$ denotes cumulative distribution function of a standard normal. In this case, the verification rate is roughly 0.47.
	\item The disease status $\mathcal{D}$, the test result $T$ and a covariate $A$ are generated in the same way as in scenario II. The verification status $V$ is simulated through the model
	\[
	\Pr(V = 1| T, A, D_1, D_2) = \frac{\exp(1 + T - 0.5A - 2D_1 - D_2)}{1 + \exp(1 + T - 0.5A - 2D_1 - D_2)}.
	\]
	The verification rate is roughly 0.42, and the true VUS still is 0.843.
	\item We generate the test result $T \sim \mathcal{U}(-3, 3)$, two covariates $A_1 \sim \mathcal{N}(0,1)$ and $A_2 \sim \mathrm{Ber}(0.6)$. The disease status $\mathcal{D} = (D_1, D_2, D_3)^\top$ is generated by using the multinomial logistic model
	\begin{eqnarray}
	\log \left(\frac{\Pr(D_1 = 1|T, A_1, A_2)}{\Pr(D_3 = 1|T, A_1, A_2)}\right) &=& 5 - 6T + 2A_1 + A_2, \nonumber \\
	\log \left(\frac{\Pr(D_2 = 1|T, A_1, A_2)}{\Pr(D_3 = 1|T, A_1, A_2)}\right) &=& 4 - 3T + 4A_1 + 2A_2. \nonumber
	\end{eqnarray}
	The verification status $V$ is simulated by the following model:
	\[
	\log \left(\frac{\Pr(V = 1| T, A_1, D_1, D_2)}{\Pr(V = 0| T, A_1, D_1, D_2)}\right) = 1 + 1.5T - A_1 + 2A_2 - 1.5D_1 - 2D_2.
	\]
	In this scenario, the verification rate is roughly $0.56$. The true VUS value is $0.74$.
	\item We generate the test result $T$, and two covariates $A_1, A_2$ as in scenario IV. The disease status $\mathcal{D} = (D_1, D_2, D_3)^\top$ is generated by using the multinomial logistic model
	\begin{eqnarray}
	\log \left(\frac{\Pr(D_1 = 1|T, A_1, A_2)}{\Pr(D_3 = 1|T, A_1, A_2)}\right) &=& 5 - 6T + 2A_1 + A_2 + A_1A_2, \nonumber \\
	\log \left(\frac{\Pr(D_2 = 1|T, A_1, A_2)}{\Pr(D_3 = 1|T, A_1, A_2)}\right) &=& 4 - 3T + 4A_1 + 2A_2 + 0.5A_1A_2. \nonumber
	\end{eqnarray}
	The verification status $V$ is simulated by the following model:
	\[
	\log \left(\frac{\Pr(V = 1| T, A_1, D_1, D_2)}{\Pr(V = 0| T, A_1, D_1, D_2)}\right) = 1 + 2T - 1.5A_1 - D_1 - 2D_2.
	\]
	The verification rate is roughly $0.46$, and the true VUS value is $0.728$.
\end{enumerate}

Scenarios I and II aim to compare the discussed mean score equation-based approach with the fully likelihood-based one, in terms of accuracy of  parameter estimators in the verification model, and VUS estimators. As pointed out above, scenario I is identical to Scenario I in \citet{toduc2019vus} and does not involve instrumental variables. On the contrary, in scenario II the covariate $A$ plays the role of  instrumental variable. Moreover, scenarios II to VI  aim to assess the finite-sample behavior of new proposed bias-corrected VUS estimators.  In particular, as explained below, scenarios IV, V and VI
serve to evaluate the effects on the estimators of some types of misspecifications in the working models.

By applying the Bayes' rule, for all considered scenarios, the true conditional (to $V=1$) disease process takes the following form:
\begin{eqnarray}
\log\left(\frac{\Pr(D_1 = 1| T, \mathbf{A}, V = 1)}{\Pr(D_3 = 1| T, \mathbf{A}, V = 1)}\right) &=& \log\left(\frac{\pi_{10}}{\pi_{00}}\right) + \log\left(\frac{\Pr(D_1 = 1| T, \mathbf{A})}{\Pr(D_3 = 1| T, \mathbf{A})}\right), \nonumber \\
\log\left(\frac{\Pr(D_2 = 1| T, \mathbf{A}, V = 1)}{\Pr(D_3 = 1| T, \mathbf{A}, V = 1)}\right) &=& \log\left(\frac{\pi_{01}}{\pi_{00}}\right) + \log\left(\frac{\Pr(D_2 = 1| T, \mathbf{A})}{\Pr(D_3 = 1| T, \mathbf{A})}\right), \nonumber
\end{eqnarray}
where $\mathbf{A}$ is simply the covariate $A$ in the scenarios I, II, IV and V,  $\mathbf{A} =(A_1, A_2)^\top$ in scenarios III and VI and $\pi_{10}$, $\pi_{01}$ and $\pi_{00}$ denote $\Pr(V = 1|T, \mathbf{A}, D_1 = 1, D_2 = 0)$, $\Pr(V = 1|T, \mathbf{A}, D_1 = 0, D_2 = 1)$ and $\Pr(V = 1|T, \mathbf{A}, D_1 = 0, D_2 = 0)$, respectively.  Clearly, in practice this is a multinomial logistic model with unknown offset terms (intercepts).

In making simulations, when considering the new mean score equation-based approach,
we fit the the conditional disease model as a multinomial logistic one, 
i.e., we specify the working conditional disease model as follows:
\begin{eqnarray}
\log\left(\frac{\Pr(D_1 = 1| T, \mathbf{A}, V = 1)}{\Pr(D_3 = 1| T, \mathbf{A}, V = 1)}\right) &=& \eta^{(1)}_0 + \eta^{(1)}_1 T + \mathbf{A}^\top \eta^{(1)}_2, \nonumber \\
\log\left(\frac{\Pr(D_2 = 1| T, \mathbf{A}, V = 1)}{\Pr(D_3 = 1| T, \mathbf{A}, V = 1)}\right) &=& \eta^{(2)}_0 + \eta^{(2)}_1 T + \mathbf{A}^\top \eta^{(2)}_2, \nonumber
\end{eqnarray}
where $\mathbf{A} = A$ for scenarios I, II and IV, $\mathbf{A} = (A_1, A_2)^\top$ for  scenarios III and V and $\mathbf{A} = (A_1^2, A_2)^\top$ for scenario VI. As for the verification model, in the estimation procedure we use the correct logistic model in scenarios I, II, III and VI, while in scenarios IV  we use a logistic models involving only $T$, and  in scenario V we use a probit model involving only $T$ and the covariate $A_1.$ Therefore, the working models are correctly specified in scenarios I, II and III,  while they present various types of misspecification in scenarios IV, V and VI. More precisely: in scenario IV, we omit the covariate $A$ in the verification model and use it as instrumental variable; in scenario V, we omit the covariate $A_2$ in the verification model, use it as instrumental variable and incorrectly define the model (probit instead of logit); in scenario VI, we use $A_1^2$ instead of $A_1$ and omit the interaction term in the conditional disease model. Finally, as for the fully likelihood-based approach, in scenarios I and II, the working (unconditional) disease model and the working verification model are both correctly specified.

\subsection{Computational aspects}
In order to obtain the new VUS estimates, in our simulation we need to solve, at each replication, the score equation $U(\boldsymbol{\eta}) = \boldsymbol{0}$, and then the empirical mean score equation $\bar{S}(\boldsymbol{\gamma}; \widehat{\boldsymbol{\eta}}) = \boldsymbol{0}$. The first  equation refers to the conditional (to $V=1$) disease model and, since we work with a multinomial logistic model in all scenarios, we resort to  the routine \texttt{multinom()} (of \texttt{R} package \texttt{nnet}) to solve it. Given $\widehat{\boldsymbol{\eta}}$, the empirical mean score equation $\bar{S}(\boldsymbol{\gamma}; \widehat{\boldsymbol{\eta}}) = \boldsymbol{0}$ could be solved in $\boldsymbol{\gamma}$ by using some numeric algorithm, such the ones deployed in the \texttt{R} routine \texttt{nleqslv()}. However, numerical algorithms for solving equations are often poorly stable and reliable, or, better, are less reliable than optimization algorithms. Therefore, in order to  obtain $\widehat{\boldsymbol{\gamma}}$, we minimize the squared Euclidean norm of $\bar{S}(\boldsymbol{\gamma}; \widehat{\boldsymbol{\eta}})$, i.e., $\left\|\bar{S}(\boldsymbol{\gamma}; \widehat{\boldsymbol{\eta}})\right\|^2$. The optimization  is  performed by using the \texttt{L-BFGS-B} algorithm supplied in the  \texttt{R} routine \texttt{optim()}. 

\subsection{Results}
\label{sec:simu:compar}
Simulation results are given in tables 1-3. Table \ref{tab:res_para} refers to the comparison between the mean score equation-based approach (MSEq) and the  fully likelihood-based approach (logLike). For three different sample sizes, i.e., 150, 250 and 500, the table provides Monte Carlo means for the estimators of the parameters in the verification model, together with Monte Carlo means for the  four (FI, MSI, IPW and PDR) VUS estimators, for scenarios I and II. Recall that, in such scenarios, the considered working models are correctly specified. The results clearly show that the new approach achieves better performance, providing less biased estimators, already at the smallest sample size. 

\begin{table}[htbp]
	\centering
	\caption{Monte Carlo means for the estimators of the parameters in the verification model and for the corresponding four bias--corrected VUS estimators. The true values are: for scenario I, $\beta_0 = 2, \beta_1 = 0.5, \beta_2 = -1.2, \lambda_1 = -2, \lambda_2 = -1$ and VUS $= 0.791$; for scenario II, $\beta_0 = 1, \beta_1 = 1, \lambda_1 = -2, \lambda_2 = -1$ and VUS $= 0.843$.} 
	\label{tab:res_para}
	\begingroup\scriptsize
	\begin{tabular}{ccrr|rr|rr}
		\toprule
		& & \multicolumn{2}{c|}{$n = 150$} & \multicolumn{2}{c|}{$n = 250$} & \multicolumn{2}{c}{$n = 500$} \\
		& & MSEq & logLike & MSEq & logLike & MSEq & logLike \\
		\midrule
		\multirow{ 9 }{*}{ Scenario I } & $\widehat{\beta}_0$ & 2.081 & 0.750 & 1.992 & 1.189 & 2.274 & 1.806 \\ 
		& $\widehat{\beta}_1$ & 0.595 & 0.782 & 0.550 & 0.658 & 0.491 & 0.546 \\ 
		& $\widehat{\beta}_2$ & $-$1.281 & $-$1.271 & $-$1.255 & $-$1.236 & $-$1.230 & $-$1.217 \\ 
		& $\widehat{\lambda}_1$ & $-$1.870 & $-$0.127 & $-$1.919 & $-$0.943 & $-$2.192 & $-$1.763 \\ 
		& $\widehat{\lambda}_2$ & $-$0.132 & $-$0.122 & $-$0.341 & $-$0.435 & $-$0.902 & $-$0.867 \\ 
		& FI & 0.775 & 0.761 & 0.776 & 0.771 & 0.784 & 0.783 \\ 
		& MSI & 0.772 & 0.757 & 0.774 & 0.769 & 0.783 & 0.783 \\ 
		& IPW & 0.778 & 0.765 & 0.776 & 0.770 & 0.783 & 0.782 \\ 
		& PDR & 0.773 & 0.757 & 0.773 & 0.765 & 0.781 & 0.782 \\ 
		\midrule
		\multirow{ 8 }{*}{ Scenario II } & $\widehat{\beta}_0$ & 3.154 & 5.071 & 2.119 & 5.217 & 1.553 & 4.011 \\ 
		& $\widehat{\beta}_1$ & 1.108 & 1.870 & 1.078 & 1.452 & 1.039 & 1.154 \\ 
		& $\widehat{\lambda}_1$ & $-$4.131 & $-$5.173 & $-$3.115 & $-$5.721 & $-$2.549 & $-$4.895 \\ 
		& $\widehat{\lambda}_2$ & $-$0.580 & $-$4.365 & $-$0.899 & $-$4.757 & $-$1.274 & $-$3.913 \\ 
		& FI & 0.844 & 0.799 & 0.845 & 0.822 & 0.843 & 0.835 \\ 
		& MSI & 0.841 & 0.795 & 0.843 & 0.820 & 0.841 & 0.833 \\ 
		& IPW & 0.841 & 0.822 & 0.842 & 0.829 & 0.841 & 0.836 \\ 
		& PDR & 0.839 & 0.801 & 0.842 & 0.825 & 0.841 & 0.835 \\ 
		\bottomrule
	\end{tabular}
	\endgroup
\end{table}

Results in Table \ref{tab:res_II_III} allow us to evaluate the behavior of the proposed VUS estimators, and of the corresponding variance estimators based on (\ref{var_est}), when the considered working models are correctly specified.
For both scenarios II and III and four different sample sizes (150, 250, 500 and 1000), the table contains Monte Carlo relative biases, Monte Carlo standard deviations and (means of) estimated standard deviations based on (\ref{var_est}), for the new VUS estimators. The table also provides the empirical coverages of the 95\% confidence intervals for the VUS, obtained through the normal approximation approach applied to each estimator and the use of the corresponding variance estimator (\ref{var_est}).

Overall, Table \ref{tab:res_II_III} shows that the proposed approach provides satisfactory results in terms of relative bias and coverage probability when the verification models are correctly specified.  At the smallest sample size, the larger absolute relative bias is equal to $1.4\%$ (scenario III). As expected, when the sample size increases, the relative bias of the  VUS estimators decreases, becoming negligible for $n \ge 250$.  Moreover, especially for scenario III (in which the verification process follows a probit model), in case of smallest sample size, $n=150$, the (means of) estimated standard deviations based on (\ref{var_est}) still seem to be affected by a poor accuracy in the estimation of the parameters in the working models. However, in similar cases, as shown in the table, more accurate estimates of standard deviations can be obtained via bootstrap. Bootstrap estimates also improve empirical coverage probabilities of confidence intervals, when used within the normal approximation approach (see also the results for scenario II). 

We emphasize that, when working models are correctly specified,  low-accuracy problems can  be related to low verification rates. For example, in scenario III the verification rate is 0.47. If we change the verification process and use the following model to simulate $V_i$ values:
\[
\Pr(V = 1| T, A_1, D_1, D_2) = \Phi(2.5 + T - 1.2A_1 - 2D_1 - D_2),
\]
the verification rate becomes roughly 0.72, and the corresponding simulation results (for $n=150$ and 1000 replications), given below,
\begin{table}[h]
	\centering
	\begingroup\footnotesize
	\begin{tabular}{lrrrr}
		\   & Bias(\%) & MCSD & ASD & CP.Asy(\%) \\
		\midrule
		FI & $-$0.6 & 0.070 & 0.075  & 95.8 \\
		MSI & $-$0.6 & 0.070 & 0.075  & 95.8 \\
		IPW & 0.1 & 0.074 & 0.069 & 93.2 \\
		PDR & $-$0.6 & 0.070 & 0.075 & 95.8 \\
		\midrule
	\end{tabular}
	\endgroup
\end{table}
are clearly better than those in Table \ref{tab:res_II_III}.
\begin{sidewaystable}
	\begin{center}
		\caption{Relative  biases (\%), Monte Carlo standard deviations (MCSD), means of estimated  standard deviation based on (\ref{var_est})  (ASD), means of  bootstrap standard deviation based on 250 bootstrap replications (BSD), empirical coverage probabilities (based on asymptotic theory and (\ref{var_est})  (CP.Asy) or asymptotic theory and bootstrap standard deviations  (CP.Bst)) of 95\% confidence intervals for the VUS  (\%). The working models are correctly specified.}
		\label{tab:res_II_III}
		\begingroup\scriptsize
		\begin{tabular}{llrrrrrr|rrrrrr}
			\toprule
			& & \multicolumn{6}{c|}{Scenario II: VUS = 0.843} & \multicolumn{6}{c}{Scenario III: VUS = 0.457} \\
			\cmidrule{3-14}
			& & Bias & MCSD & ASD & BSD & CP.Asy & CP.Bst & Bias & MCSD & ASD & BSD & CP.Asy & CP.Bst \\
			\midrule
			\multirow{4}{*}{$n = 150$} & FI & 0.2 & 0.054 & 0.055 & 0.057 & 88.9 & 92.2 & $-$1.4 & 0.070 & 0.144 & 0.075 & 96.8 & 94.5 \\ 
			& MSI & $-$0.3 & 0.055 & 0.072 & 0.059 & 90.0 & 92.8 & $-$1.4 & 0.070 & 0.144 & 0.075 & 96.9 & 94.5 \\ 
			& IPW & $-$0.3 & 0.059 & 0.064 & 0.061 & 88.9 & 91.3 & $-$0.6 & 0.083 & 0.083 & 0.087 & 93.8 & 94.3 \\ 
			& PDR & $-$0.4 & 0.061 & 0.088 & 0.073 & 89.1 & 92.8 & $-$1.4 & 0.070 & 0.144 & 0.075 & 96.9 & 94.5 \\ 
			\midrule
			\multirow{4}{*}{$n = 250$} & FI & 0.3 & 0.041 & 0.038 & -- & 91.0 & -- & $-$0.1 & 0.052 & 0.063 & -- & 96.1 & -- \\ 
			& MSI & $-$0.0 & 0.042 & 0.040 & -- & 92.3 & -- & $-$0.1 & 0.052 & 0.062 & -- & 96.1 & -- \\ 
			& IPW & $-$0.0 & 0.044 & 0.045 & -- & 91.6 & -- & $-$0.2 & 0.063 & 0.064 & -- & 94.5 & -- \\ 
			& PDR & $-$0.1 & 0.044 & 0.049 & -- & 91.2 & -- & $-$0.1 & 0.052 & 0.063 & -- & 96.2 & -- \\
			\midrule
			\multirow{4}{*}{$n = 500$} & FI & $-$0.0 & 0.028 & 0.026 & -- & 92.5 & -- & $-$0.8 & 0.037 & 0.040 & -- & 95.6 & -- \\ 
			& MSI & $-$0.2 & 0.028 & 0.028 & -- & 93.5 & -- & $-$0.8 & 0.037 & 0.039 & -- & 95.6 & -- \\ 
			& IPW & $-$0.2 & 0.030 & 0.029 & -- & 93.0 & -- & $-$0.5 & 0.043 & 0.045 & -- & 95.3 & -- \\ 
			& PDR & $-$0.2 & 0.030 & 0.028 & -- & 92.3 & -- & $-$0.8 & 0.036 & 0.039 & -- & 95.7 & -- \\ 
			\midrule
			\multirow{4}{*}{$n = 1000$} & FI & 0.1 & 0.019 & 0.019 & -- & 95.5 & -- & $-$0.2 & 0.025 & 0.029 & -- & 95.6 & -- \\ 
			& MSI & 0.0 & 0.019 & 0.020 & -- & 95.8 & -- & $-$0.2 & 0.025 & 0.029 & -- & 95.5 & -- \\ 
			& IPW & 0.0 & 0.020 & 0.021 & -- & 95.4 & -- & 0.0 & 0.031 & 0.032 & -- & 95.6 & -- \\ 
			& PDR & 0.0 & 0.020 & 0.020 & -- & 94.8 & -- & $-$0.2 & 0.025 & 0.029 & -- & 95.5 & -- \\ 
			\bottomrule
		\end{tabular}
		\endgroup
	\end{center}
\end{sidewaystable}

Table \ref{tab:res_IV_V_VI} presents the simulation results for scenarios IV-VI,  in which the working models are, in various ways,  misspecified. 
Again, for each scenario and four different sample sizes, the table gives 
Monte Carlo relative bias, Monte Carlo standard deviations and (means of) estimated standard deviations based on (\ref{var_est}), for the new VUS estimators. The table also provides the empirical coverages of the 95\% confidence intervals for the VUS, obtained through the normal approximation approach applied to each estimator and the use of the corresponding variance estimator (\ref{var_est}).

In the considered scenarios, all proposed VUS estimators appear to be biased, even when the sample size increases. However, simulation results appear to stay on acceptable levels (the relative biases are not too large and the coverage probabilities are close to 95\%) in scenarios IV and V, i.e., when the misspecification follows by the omission of a variable and/or a wrong choice of the link function in the working generalized linear model for the verification process. Clearly, due to the strong deviation between the working conditional disease model and the true one, the results are significantly worse in the scenario VI. Anyway, variance  estimators (\ref{var_est}) seem relatively more robust with respect to the considered  misspecifications in the working models.

\begin{sidewaystable}
	\begin{center}
		\caption{Relative  biases (\%), Monte Carlo standard deviations (MCSD), means of estimated  standard deviation based on (\ref{var_est})  (ASD), empirical coverage probabilities (based on asymptotic theory and (\ref{var_est})  (CP.Asy) of 95\% confidence intervals for the VUS  (\%).The working verification model (scenarios IV and V) and the working conditional disease model (scenario VI) are misspecified.}
		\label{tab:res_IV_V_VI}
		\begingroup\scriptsize
		\begin{tabular}{llrrrr|rrrr|rrrr}
			\toprule
			& & \multicolumn{4}{c|}{Scenario IV: VUS = 0.843} & \multicolumn{4}{c|}{Scenario V: VUS = 0.74} & \multicolumn{4}{c}{Scenario VI: VUS = 0.728} \\
			\cmidrule{3-14}
			& & Bias & MCSD & ASD & CP & Bias & MCSD & ASD & CP & Bias & MCSD & ASD & CP \\
			\midrule
			\multirow{4}{*}{$n = 150$} & FI & $-$0.8 & 0.059 & 0.055 & 89.9 & $-$3.6 & 0.090 & 0.063 & 89.6 & $-$5.2 & 0.081 & 0.068 & 90.7 \\ 
			& MSI & $-$1.1 & 0.059 & 0.060 & 91.4 & $-$3.7 & 0.090 & 0.062 & 89.6 & $-$5.4 & 0.082 & 0.073 & 91.8 \\ 
			& IPW & $-$0.1 & 0.064 & 0.063 & 88.9 & $-$2.2 & 0.070 & 0.066 & 95.4 & $-$5.5 & 0.083 & 0.071 & 91.1 \\ 
			& PDR & $-$1.1 & 0.067 & 0.066 & 90.3 & $-$3.4 & 0.091 & 0.079 & 90.9 & $-$3.7 & 0.097 & 0.074 & 91.2 \\ 
			\midrule
			\multirow{4}{*}{$n = 250$} & FI & $-$0.7 & 0.044 & 0.041 & 92.8 & $-$3.6 & 0.081 & 0.046 & 89.7 & $-$5.2 & 0.067 & 0.049 & 87.8 \\ 
			& MSI & $-$1.0 & 0.044 & 0.041 & 93.8 & $-$3.7 & 0.081 & 0.045 & 89.2 & $-$5.4 & 0.068 & 0.055 & 90.3 \\ 
			& IPW & $-$0.1 & 0.050 & 0.050 & 91.7 & $-$2.2 & 0.059 & 0.055 & 94.8 & $-$5.5 & 0.065 & 0.060 & 90.0 \\ 
			& PDR & $-$1.1 & 0.051 & 0.051 & 92.3 & $-$3.4 & 0.083 & 0.054 & 90.4 & $-$3.7 & 0.068 & 0.056 & 90.7 \\ 
			\midrule
			\multirow{4}{*}{$n = 500$} & FI & $-$0.9 & 0.031 & 0.028 & 92.7 & $-$1.8 & 0.050 & 0.032 & 94.4 & $-$5.0 & 0.047 & 0.036 & 81.9 \\ 
			& MSI & $-$1.1 & 0.031 & 0.030 & 94.9 & $-$1.8 & 0.050 & 0.031 & 93.6 & $-$5.2 & 0.047 & 0.040 & 87.2 \\ 
			& IPW & $-$0.1 & 0.034 & 0.032 & 91.7 & $-$1.3 & 0.039 & 0.045 & 97.1 & $-$5.4 & 0.049 & 0.040 &  85.0 \\ 
			& PDR & $-$1.1 & 0.036 & 0.033 & 92.9 & $-$1.7 & 0.053 & 0.036 & 94.1 & $-$3.4 & 0.048 & 0.035 & 87.5 \\ 
			\midrule
			\multirow{4}{*}{$n = 1000$} & FI & $-$0.8 & 0.020 & 0.020 & 94.1 & $-$1.1 & 0.023 & 0.022 & 96.2 & $-$5.2 & 0.031 & 0.025 & 69.1 \\ 
			& MSI & $-$0.9 & 0.020 & 0.021 & 95.1 & $-$1.1 & 0.023 & 0.021 & 94.9 & $-$5.3 & 0.031 & 0.029 & 78.0 \\ 
			& IPW & 0.2 & 0.022 & 0.023 & 93.9 & $-$1.0 & 0.024 & 0.035 & 97.5 & $-$5.1 & 0.034 & 0.029 & 74.8 \\ 
			& PDR & $-$0.8 & 0.023 & 0.023 & 95.4 & $-$0.9 & 0.028 & 0.024 & 95.1 & $-$3.2 & 0.032 & 0.025 & 82.0 \\ 
			\bottomrule
		\end{tabular}
		\endgroup
	\end{center}
\end{sidewaystable}

\section{Application}
\label{sec:app}
We illustrate an application of our method by using data from  the Pre-PLCO Phase II Dataset, from the SPORE/Early Detection Network/Prostate, Lung, Colon, and Ovarian Cancer Study. The study protocol and the data are publicly available at the address \footnote{\url{http://edrn.nci.nih.gov/protocols/119-spore-edrn-pre-plco-ovarian-phase-ii-validation}}. In this study, several serum protein biomarkers are evaluated in terms of their ability to correctly classify possible epithelial ovarian cancer (EOC) cases into benign disease, early stage (I-II) and late stage (III-IV). Among such biomarkers, we are interested in insulin-like growth factor-II (IGF-II), previously studied in \citet{mor2005serum} and \citet{visintin2008diagnostic} to distinguish non-diseased from cancer cases (early stage and the late stage). The data we consider refer to 156 patients with benign disease, 71 patients with early stage, and 82 with late stage. Moreover, the data present 87 unverified units; therefore,  only 78\% of patients receive the assessment of their true disease status.
\begin{figure}[htpb]
	\begin{center}
		\includegraphics[width=0.7\textwidth]{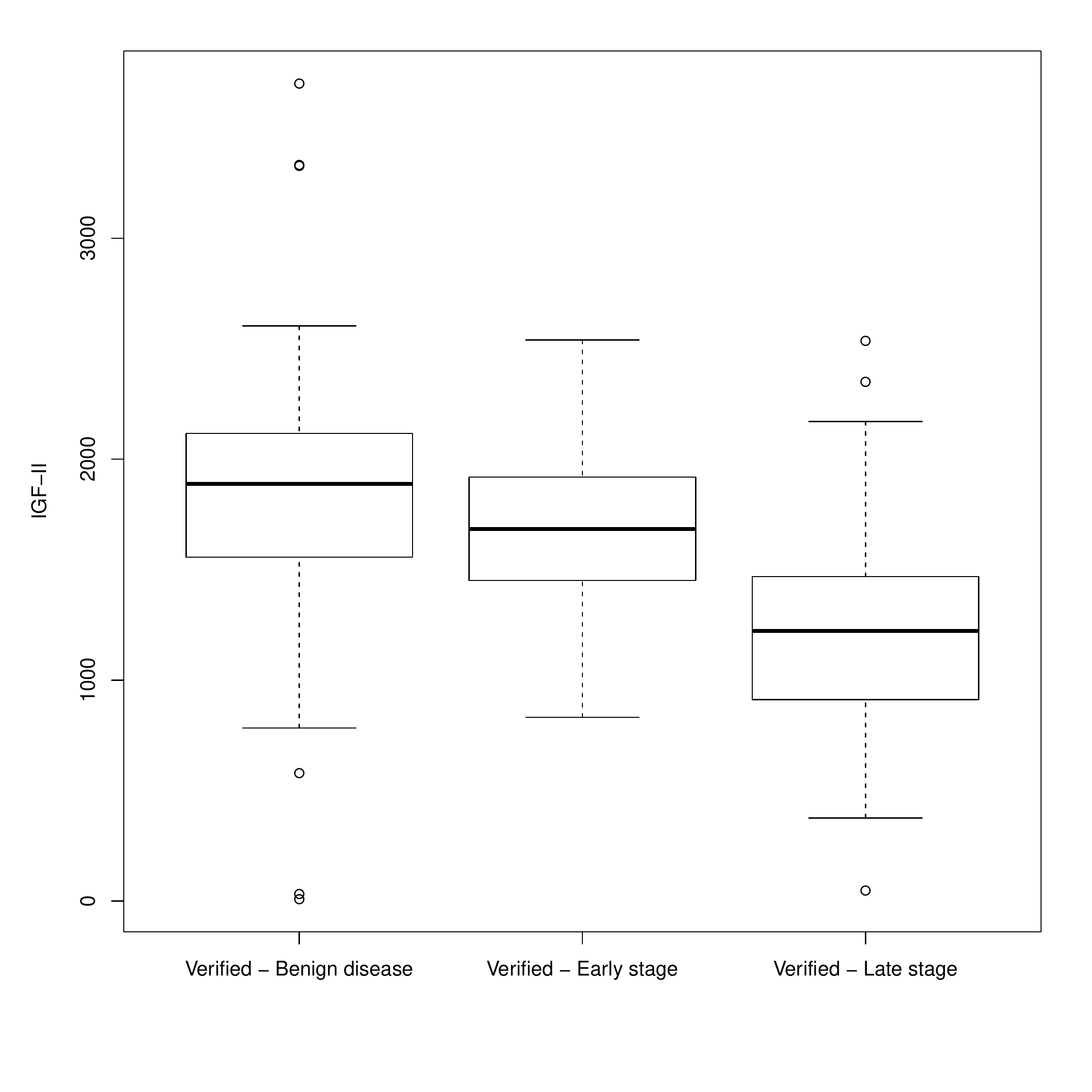}
		\caption{Boxplot for the IGF-II values, stratified by the disease status (verified units only).}
		\label{fig:igf2}
	\end{center}
\end{figure}

The boxplots for the IGF-II values, stratified by the disease status and based on the verified units only, are presented in Figure \ref{fig:igf2}, which show that lower values of IGF-II are associated with higher severity of disease; thus, in our analysis, we consider minus the values of the IGF-II as test values. Moreover, we consider two other serum proteins, i.e., HE4 and CA125, as covariates. 

Our goal is to estimate the VUS for the IGF-II biomarker. In the analysis, we standardize the values of IGF-II, HE4 and CA125, and below, for convenience, denote them as $T$, $A_1$ and $A_2$, respectively. We also let $D_1$ and $D_2$ to be the  variables corresponding to benign disease and early stage, respectively.

We fit the conditional (to $V=1$) disease model  by using a multinomial logistic regression model with $T$, $A_1$ and $A_2$ as predictors and late stage as reference level. As for the verification model, we consider a logistic model and a probit model, with only $T$ as a predictor; hence, $A_1$ and $A_2$ are used as instrumental variables. The estimated parameters, estimated  standard errors and corresponding $p$-values for the significance test for coefficients in the models are given in Table \ref{tab:eoc_para}. The $p$-values refer to the Wald-type statistics, which use results in (\ref{est:Sigma:gamma}). The resulting $p$-values for ignorable/nonignorable coefficients in both models (logistic and probit), indicate that, in such case, the missing data mechanism is nonignorable.

\begin{table}[htbp]
	\begin{center}
		\caption{Application to IGF-II biomarker: estimated coefficients (Est.), standard errors (SE) and corresponding $p$-values 
			for (conditional) disease and  verification models; VUS estimates, along with related standard errors and 95\% confidence intervals.}
		\label{tab:eoc_para}
		\begin{scriptsize}
			\begin{tabular}{c r r r | r r r }
				\toprule
				& \multicolumn{6}{c}{Disease model} \\
				\cmidrule(l){2-7}
				& \multicolumn{3}{c|}{Benign disease} & \multicolumn{3}{c}{Early stage} \\
				\cmidrule(l){2-7}
				& Est. & SE & $p$-value & Est. & SE & $p$-value \\ 
				\midrule
				Intercept & $-$1.832 & 0.786 & 0.020 & 0.528 & 0.218 & 0.015 \\ 
				$T$ & $-$0.922 & 0.204 & $<$ 0.001 & $-$0.880 & 0.205 & $<$ 0.001 \\
				$A_1$ & $-$1.289 & 0.444 & 0.004 & $-$0.625 & 0.168 & $<$ 0.001 \\
				$A_2$ & $-$17.131 & 4.325 & $<$ 0.001 & $-$0.026 & 0.141 & 0.852 \\
				\toprule
				& \multicolumn{6}{c}{Verification model} \\
				\cmidrule(l){2-7}
				& \multicolumn{3}{c|}{Logit link} & \multicolumn{3}{c}{Probit link} \\
				\cmidrule(l){2-7}
				& Est. & SE & $p$-value & Est. & SE & $p$-value \\ 
				Intercept & 6.186 & 1.455 & $<$0.001 & 2.966 & 0.621 & $<$0.001\\
				$T$ & $-$1.338 & 0.376 & $<$0.001 & $-$0.659 & 0.162 & $<$0.001\\
				$D_1$ & $-$5.733 & 1.501 & $<$0.001 & $-$2.649 & 0.639 & $<$0.001\\
				$D_2$ & 4.207 & 0.961 & $<$0.001 & 1.899 & 0.631 & 0.003\\
				\toprule
				& \multicolumn{6}{c}{bias--corrected VUS} \\
				\cmidrule(l){2-7}
				& \multicolumn{3}{c|}{Logit link} & \multicolumn{3}{c}{Probit link} \\
				\cmidrule(l){2-7}
				& Est. & SE & 95\% CI & Est. & SE & 95\% CI \\
				\cmidrule(l){2-7}
				FI & 0.317 & 0.025 & $(0.268, 0.366)$ & 0.317 & 0.025 & $(0.268, 0.366)$\\
				MSI & 0.343 & 0.028 & $(0.287, 0.398)$ & 0.343 & 0.028 & $(0.287, 0.398)$\\
				IPW & 0.243 & 0.087 & $(0.073, 0.413)$ & 0.299 & 0.056 & $(0.188, 0.410)$\\
				PDR & 0.342 & 0.027 & $(0.289, 0.395)$ & 0.341 & 0.030 & $(0.283, 0.399)$\\
				\bottomrule
			\end{tabular} 
		\end{scriptsize}
	\end{center}
\end{table}

Table \ref{tab:eoc_para} also provides the four bias-corrected VUS estimates, along with  the  standard errors and 95\% confidence intervals (CI), constructed using normal approximation and (\ref{var_est}). The Na\"{i}ve VUS estimate, calculated from verified subjects only, is 0.460 with standard error 0.035 and  95\% CI $(0.393,  0.528)$. Therefore, all bias-corrected VUS estimates are significantly different to the Na\"{i}ve one. Moreover, unlike the logistic model, the IPW estimate in probit model is comparable to the other VUS estimates. This may suggest that the probit link function is a good choice for the generalized regression model describing the verification process for these data. 

Finally, we highlight the purely illustrative nature of this application and the fact that, in real applications, the approach proposed in this paper allows to evaluate the model chosen to describe the (conditional) disease process by diagnostic procedures and/or {\it ad hoc} statistical tests present in the litterature,  as those, for instance, in \citet{goeman2006goodness} for the multinomial logistic regression model. 

\section{Conclusion}
\label{sec:concl}
In this paper, we propose a mean score equation-based approach to estimate the volume under a ROC surface when the disease status is missing not at random. This approach can  take advantage of the use of instrumental variables, which help avoid possible identifiability problems. We prove that our proposed bias--corrected VUS estimators are consistent and asymptotically normal. 

Our method is essentially a parametric method, and requires the estimation of the parameters of a (parametric) model chosen to describe the conditional (to $V=1$) disease process, and then the solution of an empirical mean score equation resulting from a parametric model chosen for the verification process. The method may also be used in a semiparametric context, by resorting to a nonparametric regression approach to fit the conditional disease model, as in \citet{morikawa2017semiparametric}. Clearly, this topic deserves further investigation.

Compared to the  fully likelihood-based approach discussed in \citet{toduc2019vus}, the new method is simpler, and more flexible  also thanks to the possible use of instrumental variables. Moreover, as our simulation results show, the new bias--corrected VUS estimators are generally more accurate with moderate sample sizes.

Resorting to instrumental variables may be necessary when relying on working models other than those, namely logistic and multinomial logistic, assumed by \citet{toduc2019vus}  to describe, respectively, the verification and  disease process. In practice, however, the choice of variables that could play the role of instrumental variables is not trivial. With respect to this issue,  we suggest a possible selection strategy that foresees three steps:
\begin{enumerate}[(i)]
	\item choose a working model for the conditional (to $V=1$) disease process and select (by using a standard  backward stepwise regression) a vector of covariates $\mathbf{A}$, statistically significant;
	\item choose a working model for the verification process and, given the estimate $\widehat{\boldsymbol{\eta}}$ obtained in step (i), use, one at a time,  the elements of $\mathbf{A}$,  jointly with the test $T$,  to solve the resulting empirical mean score equation 
	$\bar{S}(\boldsymbol{\gamma}; \widehat{\boldsymbol{\eta}}) = \boldsymbol{0}$;
	\item  take as vector of instrumental variables, the vector $\mathbf{A}_2$ consisting of the elements of $\mathbf{A}$ resulting statistically non-significant at the previous step (ii).
\end{enumerate}
Then, the final estimate $\widehat{\boldsymbol{\gamma}}$ is obtained by solving the empirical mean score equation based on $T$ and, when $\mathbf{A}_2$ does not coincide with $\mathbf{A}$, the component $\mathbf{A}_1$  of $\mathbf{A} = \left(\mathbf{A}^\top_1, \mathbf{A}^\top_2\right)^\top$. This is also the selection strategy  that we used in the illustrative application of Section \ref{sec:app}.

\appendix
\section{Conditional probabilities of disease for unverified subjects}
\label{app:A}
Applying Bayes' rule, we have
\[
\Pr(D_k = 1|V = v, T, \mathbf{A}) = \frac{\Pr(V = v|T, \mathbf{A}_1, D_k = 1)\Pr(D_k = 1| T, \mathbf{A})}{\Pr(V = v|T, \mathbf{A})}
\]
for $v = 0, 1$. Thus, we obtain the following ratios 
\begin{eqnarray}
\frac{\Pr(D_1 = 1|V = 1, T, \mathbf{A})}{\Pr(D_1 = 1|V = 0, T, \mathbf{A})} \bigg / \frac{\Pr(D_3 = 1|V = 1, T, \mathbf{A})}{\Pr(D_3 = 1|V = 0, T, \mathbf{A})} &=& \frac{\mathrm{Odd} (T, \mathbf{A}_1, 1,0; \boldsymbol{\gamma})}{\mathrm{Odd} (T, \mathbf{A}_1, 0,0; \boldsymbol{\gamma})}, \nonumber \\ 
\frac{\Pr(D_2 = 1|V = 1, T, \mathbf{A})}{\Pr(D_2 = 1|V = 0, T, \mathbf{A})} \bigg / \frac{\Pr(D_3 = 1|V = 1, T, \mathbf{A})}{\Pr(D_3 = 1|V = 0, T, \mathbf{A})} &=& \frac{\mathrm{Odd} (T, \mathbf{A}_1, 0,1; \boldsymbol{\gamma})}{\mathrm{Odd} (T, \mathbf{A}_1, 0,0; \boldsymbol{\gamma})}, \nonumber 
\end{eqnarray}
where
\[
\mathrm{Odd} (T, \mathbf{A}_1, d_1, d_2; \boldsymbol{\gamma}) = \frac{\pi(T, \mathbf{A}_1, d_1, d_2; \boldsymbol{\gamma})}{1 - \pi(T, \mathbf{A}_1, d_1, d_2; \boldsymbol{\gamma})}
\]
with $(d_1, d_2) \in \left\{(1,0), (0,1), (0,0)\right\}$. These expressions are equivalent to the system of equations
\begin{subequations}
	\begin{align}
	\frac{\rho_{1(1)}}{\rho_{3(1)}} &= \frac{\rho_{1(0)}}{\rho_{3(0)}} \frac{\mathrm{Odd} (T, \mathbf{A}_1, 1, 0; \boldsymbol{\gamma})}{\mathrm{Odd} (T, \mathbf{A}_1, 0,0; \boldsymbol{\gamma})}, \label{sub_eq:1}\\
	\frac{\rho_{2(1)}}{\rho_{3(1)}} &= \frac{\rho_{2(0)}}{\rho_{3(0)}} \frac{\mathrm{Odd} (T, \mathbf{A}_1, 0, 1; \boldsymbol{\gamma})}{\mathrm{Odd} (T, \mathbf{A}_1, 0,0; \boldsymbol{\gamma})}. \label{sub_eq:2}
	\end{align}
\end{subequations}
By using the fact that $\rho_{3(0)} = 1 - \rho_{1(0)} - \rho_{2(0)}$, from  (\ref{sub_eq:1}), we obtain
\begin{equation}
\rho_{1(0)} = \frac{\rho_{1(1)}(1 - \rho_{2(0)})}{\rho_{3(1)}\frac{\mathrm{Odd} (T, \mathbf{A}_1, 1, 0; \boldsymbol{\gamma})}{\mathrm{Odd} (T, \mathbf{A}_1, 0,0; \boldsymbol{\gamma})} + \rho_{1(0)}}
\label{sub_eq:3} 
\end{equation}
Substituting (\ref{sub_eq:3}) into (\ref{sub_eq:2}), leads to the following equation
\[
\frac{\rho_{2(1)}}{\rho_{3(1)}} = \frac{\rho_{2(0)}}{1 - \frac{\rho_{1(1)}(1 - \rho_{2(0)})}{\rho_{3(1)}\frac{\mathrm{Odd} (T, \mathbf{A}_1, 1, 0; \boldsymbol{\gamma})}{\mathrm{Odd} (T, \mathbf{A}_1, 0,0; \boldsymbol{\gamma})} + \rho_{1(0)}} - \rho_{2(0)}} \frac{\mathrm{Odd} (T, \mathbf{A}_1, 0, 1; \boldsymbol{\gamma})}{\mathrm{Odd} (T, \mathbf{A}_1, 0,0; \boldsymbol{\gamma})}.
\]
After some algebra, we obtain that $\rho_{2(0)}$ is equal to
\[
\frac{\rho_{2(1)} \frac{\mathrm{Odd} (T, \mathbf{A}_1, 1, 0; \boldsymbol{\gamma})}{\mathrm{Odd} (T, \mathbf{A}_1, 0,0; \boldsymbol{\gamma})}}{\rho_{1(1)} \frac{\mathrm{Odd} (T, \mathbf{A}_1, 0, 1; \boldsymbol{\gamma})}{\mathrm{Odd} (T, \mathbf{A}_1, 0,0; \boldsymbol{\gamma})} + \rho_{2(1)} \frac{\mathrm{Odd} (T, \mathbf{A}_1, 1, 0; \boldsymbol{\gamma})}{\mathrm{Odd} (T, \mathbf{A}_1, 0,0; \boldsymbol{\gamma})} + \rho_{3(1)} \frac{\mathrm{Odd} (T, \mathbf{A}_1, 1, 0; \boldsymbol{\gamma})}{\mathrm{Odd} (T, \mathbf{A}_1, 0, 0; \boldsymbol{\gamma})}\frac{\mathrm{Odd} (T, \mathbf{A}_1, 0, 1; \boldsymbol{\gamma})}{\mathrm{Odd} (T, \mathbf{A}_1, 0,0; \boldsymbol{\gamma})}},
\]
and then $\rho_{1(0)}$ is
\[
\frac{\rho_{1(1)} \frac{\mathrm{Odd} (T, \mathbf{A}_1, 0, 1; \boldsymbol{\gamma})}{\mathrm{Odd} (T, \mathbf{A}_1, 0,0; \boldsymbol{\gamma})}}{\rho_{1(1)} \frac{\mathrm{Odd} (T, \mathbf{A}_1, 0, 1; \boldsymbol{\gamma})}{\mathrm{Odd} (T, \mathbf{A}_1, 0,0; \boldsymbol{\gamma})} + \rho_{2(1)} \frac{\mathrm{Odd} (T, \mathbf{A}_1, 1, 0; \boldsymbol{\gamma})}{\mathrm{Odd} (T, \mathbf{A}_1, 0,0; \boldsymbol{\gamma})} + \rho_{3(1)} \frac{\mathrm{Odd} (T, \mathbf{A}_1, 1, 0; \boldsymbol{\gamma})}{\mathrm{Odd} (T, \mathbf{A}_1, 0, 0; \boldsymbol{\gamma})}\frac{\mathrm{Odd} (T, \mathbf{A}_1, 0, 1; \boldsymbol{\gamma})}{\mathrm{Odd} (T, \mathbf{A}_1, 0,0; \boldsymbol{\gamma})}}
\]
In case of the logistic model (\ref{veri:model:2}), it is straightforward to show that
\[
\frac{\mathrm{Odd} (T, \mathbf{A}_1, 1, 0; \boldsymbol{\gamma})}{\mathrm{Odd} (T, \mathbf{A}_1, 0,0; \boldsymbol{\gamma})} = \exp(\lambda_1), \qquad \frac{\mathrm{Odd} (T, \mathbf{A}_1, 0, 1; \boldsymbol{\gamma})}{\mathrm{Odd} (T, \mathbf{A}_1, 0,0; \boldsymbol{\gamma})} = \exp(\lambda_2).
\]
Thus, we get the results in (\ref{expr:rho_0}).

\section{Conditions for asymptotic normality of $\left(\widehat{\boldsymbol{\eta}}^\top, \widehat{\boldsymbol{\gamma}}^\top\right)^\top$}

Recall that  $\mathcal{I}_{\boldsymbol{\eta}}(\boldsymbol{\eta}) = \E \left\{- \frac{\partial u_i(\boldsymbol{\eta})}{\partial \boldsymbol{\eta}^\top} \right\}$, $\mathcal{I}_{s,\boldsymbol{\gamma}}(\boldsymbol{\eta}, \boldsymbol{\gamma}) = \E \left\{- \frac{\partial s_i(\boldsymbol{\gamma}; \boldsymbol{\eta})}{\partial \boldsymbol{\gamma}^\top} \right\}$, $\mathcal{I}_{s,\boldsymbol{\eta}}(\boldsymbol{\eta}, \boldsymbol{\gamma}) = \E \left\{- \frac{\partial s_i(\boldsymbol{\gamma}; \boldsymbol{\eta})}{\partial \boldsymbol{\eta}^\top}
\right\}$ and consider the conditions: 
\label{app:B}
\begin{enumerate}[(D1)]
	\item  $\Pr \left(D_k = 1| T, \mathbf{A}_{1}, \mathbf{A}_{2} = \boldsymbol{a}_{2}^{(1)} \right) \ne \Pr \left(D_k = 1| T, \mathbf{A}_{1}, \mathbf{A}_{2} = \boldsymbol{a}_{2}^{(2)} \right)$ for given two values $\boldsymbol{a}_{2}^{(1)} \ne \boldsymbol{a}_{2}^{(2)}$ of the (instrumental) variable $\mathbf{A}_{2}$, and
	\[\Pr(V = 1|T, \mathbf{A}, D_1, D_2) = \Pr(V = 1|T, \mathbf{A}_1, D_1, D_2) = \pi(T, \mathbf{A}_1, D_1, D_2; \boldsymbol{\gamma}_0)\] 
	with $\mathbf{A} = \left(\mathbf{A}^\top_1, \mathbf{A}^\top_2\right)^\top$;
	\item $\pi(T_i, \mathbf{A}_{1i}, D_{1i}, D_{2i}; \boldsymbol{\gamma})$ is a continuous function of $\boldsymbol{\gamma}$ with first and second derivatives continuous in an open set containing $\boldsymbol{\gamma}_0$ as an interior point;
	\item the variables  $(T_i, \mathbf{A}^\top_i, D_{1i}, D_{2i}, D_{3i}, V_i)^\top$, $i=1,\ldots, n$, are independent and identically distributed;
	\item the parameter spaces $\Theta$ for $\boldsymbol{\eta}$ and $\Omega$ for $\boldsymbol{\gamma}$ are compact and have finite dimension;
	\item $\pi(T_i, \mathbf{A}_{1i}, D_{1i}, D_{2i}; \boldsymbol{\gamma})$ is bounded away from 0.
	\item the MLE , $\widehat{\boldsymbol{\gamma}}_n$ say,
	is consistent and asymptotically normal;
	\item $\E \left\{ U(\boldsymbol{\eta}) \right\} = \boldsymbol{0}$ has a unique solution $\boldsymbol{\eta}_0 \in \Theta$;
	\item $\mathcal{I}_{\boldsymbol{\eta}}(\boldsymbol{\eta}_0)$ exists and is invertible;
	\item there exists a neighborhood $\mathcal{N}$ of $\boldsymbol{\eta}_0$ such that the quantiles $\sup_{\boldsymbol{\eta} \in \mathcal{N}} \|U(\boldsymbol{\eta})\|$, $\sup_{\boldsymbol{\eta} \in \mathcal{N}} \|\frac{\partial}{\partial \boldsymbol{\eta}^\top} U(\boldsymbol{\eta})\|$ and $\sup_{\boldsymbol{\eta} \in \mathcal{N}} \|U(\boldsymbol{\eta})U(\boldsymbol{\eta})^\top\|$ have finite expectations, with $\| \mathbf{X}\| \equiv \sum_{i}\sum_{j}X_{ij}^2$.
\end{enumerate} 
By Theorem 1 in \citet{riddles2016propensity}, under the regularity conditions (D1)--(D9), the solution $\left(\widehat{\boldsymbol{\eta}}^\top, \widehat{\boldsymbol{\gamma}}^\top\right)^\top$ of
\[
\begin{pmatrix}
U(\boldsymbol{\eta}) \\ \bar{S}(\boldsymbol{\gamma};\boldsymbol{\eta})
\end{pmatrix} = \begin{pmatrix}
\boldsymbol{0} \\ \boldsymbol{0}
\end{pmatrix}
\]
is a consistent estimator of $\left(\boldsymbol{\eta}^\top_0, \boldsymbol{\gamma}^\top_0\right)^\top$, and satisfies
\begin{equation}
\sqrt{n} \begin{pmatrix}
\widehat{\boldsymbol{\eta}} - \boldsymbol{\eta}_0 \\
\widehat{\boldsymbol{\gamma}} - \boldsymbol{\gamma}_0
\end{pmatrix} = \begin{pmatrix}
\mathcal{I}_{\boldsymbol{\eta}}(\boldsymbol{\eta}_0) & \boldsymbol{0} \\
\mathcal{I}_{s,\boldsymbol{\eta}}(\boldsymbol{\eta}_0, \boldsymbol{\gamma}_0) & \mathcal{I}_{s,\boldsymbol{\gamma}}(\boldsymbol{\eta}_0, \boldsymbol{\gamma}_0)
\end{pmatrix}^{-1} \frac{1}{\sqrt{n}}\sum_{i = 1}^{n} \begin{pmatrix}
V_iu_i(\boldsymbol{\eta}_0) \\ s_i(\boldsymbol{\gamma}_0; \boldsymbol{\eta}_0)
\end{pmatrix} + o_p(1).
\label{asy:gamma:eta}
\end{equation}
Thus, by the Central Limit Theorem, $\left(\widehat{\boldsymbol{\eta}}^\top, \widehat{\boldsymbol{\gamma}}^\top\right)^\top$ is asymptotically normal. Moreover, 
\begin{equation}
\sqrt{n}\left(\widehat{\boldsymbol{\gamma}} - \boldsymbol{\gamma}_0\right) \stackrel{d}{\longrightarrow} \mathcal{N}\left(\boldsymbol{0}, \Sigma_{\boldsymbol{\gamma}}\right),
\label{gamma:asy}
\end{equation}
where
\[
\Sigma_{\boldsymbol{\gamma}} = \mathcal{I}^{-1}_{s,\boldsymbol{\gamma}}(\boldsymbol{\eta}_0, \boldsymbol{\gamma}_0) \mathrm{Cov}\left\{s_i(\boldsymbol{\gamma}_0; \boldsymbol{\eta}_0) - \mathcal{I}_{s,\boldsymbol{\eta}}(\boldsymbol{\eta}_0, \boldsymbol{\gamma}_0) \mathcal{I}^{-1}_{\boldsymbol{\eta}}(\boldsymbol{\eta}_0) u_i(\boldsymbol{\eta}_0) \right\} \mathcal{I}^{-1}_{s,\boldsymbol{\gamma}}(\boldsymbol{\eta}_0, \boldsymbol{\gamma}_0)
\]
and
\begin{equation}
\sqrt{n}\left(\widehat{\boldsymbol{\eta}} - \boldsymbol{\eta}_0\right) \stackrel{d}{\longrightarrow} \mathcal{N}\left(\boldsymbol{0}, \Sigma_{\boldsymbol{\eta}}\right),
\end{equation}
where $\Sigma_{\boldsymbol{\eta}} =  \mathcal{I}^{-1}_{\boldsymbol{\eta}}(\boldsymbol{\eta}_0)\mathrm{Cov}(u_i(\boldsymbol{\eta}_0)) \mathcal{I}^{-1}_{\boldsymbol{\eta}}(\boldsymbol{\eta}_0)$. Consistent estimators $\widehat{\Sigma}_{\boldsymbol{\gamma}}$ of $\Sigma_{\boldsymbol{\gamma}}$ and $\widehat{\Sigma}_{\boldsymbol{\eta}}$ of $\Sigma_{\boldsymbol{\eta}}$ can be obtained by
\begin{eqnarray}
\widehat{\Sigma}_{\boldsymbol{\eta}} &=& n \left\{\sum_{i = 1}^{n}\frac{\partial}{\partial \boldsymbol{\eta}^\top} V_iu_i(\widehat{\boldsymbol{\eta}})\right\}^{-1} \left\{\sum_{i = 1}^{n} V_iu_i(\widehat{\boldsymbol{\eta}}) u_i(\widehat{\boldsymbol{\eta}})^{\top} \right\} \left\{\sum_{i = 1}^{n}\frac{\partial}{\partial \boldsymbol{\eta}^\top} V_iu_i(\widehat{\boldsymbol{\eta}})\right\}^{-\top}, \nonumber \\
\widehat{\Sigma}_{\boldsymbol{\gamma}} &=& n\left\{\sum_{i = 1}^{n} \frac{\partial}{\partial \boldsymbol{\gamma}^\top} s_i(\widehat{\boldsymbol{\gamma}}; \widehat{\boldsymbol{\eta}})\right\}^{-1} \left\{\sum_{i = 1}^{n} f_i(\widehat{\boldsymbol{\eta}}, \widehat{\boldsymbol{\gamma}}) f_i(\widehat{\boldsymbol{\eta}}, \widehat{\boldsymbol{\gamma}})^{\top} \right\} \left\{\sum_{i = 1}^{n} \frac{\partial}{\partial \boldsymbol{\gamma}^\top} s_i(\widehat{\boldsymbol{\gamma}}; \widehat{\boldsymbol{\eta}})\right\}^{-\top}, \label{est:Sigma:gamma}
\end{eqnarray}
with $f_i(\widehat{\boldsymbol{\eta}}, \widehat{\boldsymbol{\gamma}}) = s_i(\widehat{\boldsymbol{\gamma}}; \widehat{\boldsymbol{\eta}}) - \left\{\sum_{i = 1}^{n} \frac{\partial}{\partial \boldsymbol{\eta}^\top} s_i(\widehat{\boldsymbol{\gamma}}; \widehat{\boldsymbol{\eta}})\right\} \left\{\sum_{i = 1}^{n}\frac{\partial}{\partial \boldsymbol{\eta}^\top} V_iu_i(\widehat{\boldsymbol{\eta}})\right\}^{-1} V_i u_i(\widehat{\boldsymbol{\eta}})$.

\section{Estimation Procedure}
\label{app:C}
\tikzstyle{decision} = [diamond, draw, text width=2cm, text badly centered, inner sep=0pt]
\tikzstyle{block} = [rectangle, draw, text width=4cm, text centered, rounded corners, minimum height=2em]
\tikzstyle{line} = [draw, -latex']
\tikzstyle{cloud} = [draw, ellipse, text width=3cm, minimum height=2em, text centered]
\begin{center}
	\begin{scriptsize}
		\begin{tikzpicture}[node distance=2cm, align=center]
		\node [block, text width=4cm] (verification) {Modelling Verification Process};
		\node [cloud, text width=4cm, left of=verification, below of=verification, xshift = -4.2cm, yshift = 0.5cm] (disease) {Modelling Disease Process for Verified Subjects (eventually using instrumental veriables)};
		\node [cloud, text width=0.5cm, below of=disease, yshift = 0.5cm] (fitdisease1) {$\widehat{\boldsymbol{\eta}}$};
		\node [cloud, text width=1.8cm, right of=fitdisease1, below of=disease, xshift=0.2cm, yshift = 0.5cm] (fitdisease) {${\rho}_{k(0)i}(\widehat{\boldsymbol{\eta}}, \lambda_1, \lambda_2)$};
		\node [block, text width=4cm, right of=fitdisease, below of=disease, xshift = 4.2cm, yshift = 0.5cm] (meanscore) {Empirical Mean Score Equation}; 
		\node [block, text width=2.5cm, below of=meanscore] (estimates) {$\widehat{\beta}_0$, $\widehat{\beta}_1$, $\widehat{\boldsymbol{\beta}}_{2,1}$, $\widehat{\lambda}_1$, $\widehat{\lambda}_2$};
		\node [block, text width=1.5cm, left of=estimates, below of=fitdisease1,xshift=2cm] (rho1) {$
			{\rho}_{k(1)i}(\widehat{\boldsymbol{\eta}})$};
		\node [block, text width=2cm, left of=estimates, below of=estimates,xshift=-1.5cm] (rho0) {$
			{\rho}_{k(0)i}(\widehat{\boldsymbol{\eta}}, \widehat{\lambda}_1, \widehat{\lambda}_2)$};
		\node [block, text width=0.5cm, right of=estimates, below of=estimates,xshift=-0.5cm] (pi) {$\widehat{\pi}_i$};
		\node [decision, below of=rho0,yshift = 0cm, xshift = -2.7cm] (fi) {VUS-FI Estimator};
		\node [decision, below of=rho0,yshift = 0cm] (msi) {VUS-MSI Estimator};
		\node [decision,below of=pi, left of=pi,yshift= 0cm,xshift= -0.3cm] (pdr) {VUS-PDR Estimator};
		\node [decision,below of=pi, right of=pi,yshift= 0cm,xshift= -1.5cm] (ipw) {VUS-IPW Estimator};
		\path [line] (verification) -- (meanscore);
		\path [line] (disease) -- (fitdisease1);
		\path [line] (fitdisease1) -- (fitdisease);
		\path [line] (fitdisease) -- (meanscore);
		\path [line] (meanscore) -- (estimates);
		\path [line] (fitdisease1) -- (rho1);
		\path [line] (estimates) -| (rho0);
		\path [line] (estimates) -| (pi);
		\path [line] (rho1) -- (fi);
		\path [line] (rho0) -| (fi);
		\path [line] (rho0) -- (msi);
		\path [line] (rho0) -| (pdr);
		\path [line] (pi) -| (ipw);	
		\path [line] (pi) -| (pdr);
		\end{tikzpicture}	
	\end{scriptsize}
\end{center}

\bibliographystyle{apalike}
\bibliography{refs_VUS_NI_IV}
\end{document}